\documentclass[preprint,showpacs,preprintnumbers,amsmath,amssymb]{revtex4}
\usepackage{amsmath}

\usepackage{graphicx}
\usepackage{pdfpages}
\usepackage{dcolumn}
\usepackage{bm}
\usepackage{amsmath} 
\usepackage{amsfonts}
\usepackage{amssymb}
\usepackage{url}
\usepackage{microtype}
\usepackage{graphicx}%

\newtheorem{theorem}{Theorem}[section]
\newtheorem{lemma}[theorem]{Lemma}
\newtheorem{proposition}[theorem]{Proposition}
\newtheorem{corollary}[theorem]{Corollary}

\newenvironment{proof}[1][Proof]{\begin{trivlist}
\item[\hskip \labelsep {\bfseries #1}]}{\end{trivlist}}

\newcommand{\qed}{\nobreak \ifvmode \relax \else
      \ifdim\lastskip<1.5em \hskip-\lastskip
      \hskip1.5em plus0em minus0.5em \fi \nobreak
      \vrule height0.75em width0.5em depth0.25em\fi}

\renewcommand{\vec}[1]{\mathbf{#1}}

\begin{document}

\preprint{}
\title{Extensions of Generalized Two-Qubit Separability Probability Analyses to Higher Dimensions, Additional Measures and New Methodologies}
\author{Paul B. Slater}
 \email{slater@kitp.ucsb.edu}
\affiliation{%
Kavli Institute for Theoretical Physics, University of California, Santa Barbara, CA 93106-4030\\
}
\date{\today}
            
\begin{abstract}
We first seek the rebit-retrit counterpart to the (formally proven by Lovas and Andai) two-rebit Hilbert-Schmidt separability probability of $\frac{29}{64} =\frac{29}{2^6} \approx 0.453125$ and the qubit-qutrit analogue of the (strongly supported) value of  
$\frac{8}{33} = \frac{2^3}{3 \cdot 11} \approx 0.242424$. We advance the possibilities of  a rebit-retrit value of 
$\frac{860}{6561} =\frac{2^2 \cdot 5 \cdot 43}{3^8} \approx 0.131078$ and a  qubit-qutrit one of 
$\frac{27}{1000} = (\frac{3}{10})^3  =\frac{3^3}{2^3 \cdot 5^3} = 0.027$. These four values for $2 \times m$ systems ($m=2,3$) suggest certain numerator/denominator sequences involving powers of $m$, which
we further investigate for $m>3$. Additionally, we find that the Hilbert-Schmidt separability/PPT-probabilities for the two-rebit, rebit-retrit and two-retrit $X$-states all equal $\frac{16}{3 \pi^2} \approx 0.54038$, as well as more generally,
that the probabilities based on induced measures are equal across these three sets. Then, we extend the master Lovas-Andai formula to  induced measures. For instance, the  two-qubit function ($k=0$) is $\tilde{\chi}_{2,0}(\varepsilon)=\frac{1}{3} \varepsilon^2 (4 -\varepsilon^2)$, yielding $\frac{8}{33}$, while its $k=1$ induced measure counterpart is
$\tilde{\chi}_{2,1}(\varepsilon)=\frac{1}{4} \varepsilon ^2 \left(3-\varepsilon ^2\right)^2$,  yielding $\frac{61}{143} =\frac{61}{11 \cdot 13} \approx 0.426573$, where $\varepsilon$ is a singular-value ratio. Interpolations between Hilbert-Schmidt and operator monotone (Bures, $\sqrt{x}$) measures are also studied. Using a recently-developed golden-ratio-related (quasirandom sequence) approach, current (significant digits) estimates of the two-rebit and two-qubit Bures separability probabilities are 0.15709 and 0.07331, respectively.
\end{abstract}


\maketitle
\tableofcontents
\section{Introduction} \label{Introduction}
It has now been formally proven by Lovas and Andai \cite[Thm. 2]{lovas2017invariance} that the separability probability with respect to Hilbert-Schmidt  measure \cite{zyczkowski2003hilbert} \cite[sec. 13.3]{bengtsson2017geometry} of the 9-dimensional convex set of 
two-rebit states is $\frac{29}{64}$. Additionally, the multifacted evidence  \cite{slater2017master,milz2014volumes,fei2016numerical,shang2015monte,slater2013concise,slater2012moment,slater2007dyson}--including a recent ``master'' extension \cite{slater2017master} of the Lovas-Andai framework to {\it generalized} two-qubit states--is strongly compelling that the corresponding value for the 15-dimensional convex set of two-qubit states is $\frac{8}{33}$ (with that of the 27-dimensional convex set of two-quater[nionic]bits being $\frac{26}{323}$ [cf. \cite{adler1995quaternionic}], among other still higher-dimensional companion random-matrix related results). (Certainly, one can, however, still aspire to a yet greater
``intuitive'' understanding of these assertions, particularly in some ``geometric/visual'' sense [cf. \cite{szarek2006structure,samuel2018lorentzian,avron2009entanglement,braga2010geometrical,gamel2016entangled,jevtic2014quantum}], as well as further formalized proofs.) It is of interest to compare/contrast these finite-dimensional studies with those other quantum-information-theoretic ones, presented in the recent comprehensive volume of Aubrun and Szarek
\cite{aubrun2017alice}, employing {\it asymptotic geometric analysis}.

By a separability probability, we mean the ratio of the volume of the separable states to the volume  of all (separable and entangled) states with respect to the chosen measure, as proposed, apparently first, by {\.Z}yczkowski, Horodecki, Sanpera and Lewenstein \cite{zyczkowski1998volume} (cf. \cite{petz1996geometries,e20020146,singh2014relative,batle2014geometric}). 

In these regards, we present the formulas derived by {\.Z}yczkowski and Sommers
for the {\it total} ($N^2-1$)-dimensional Hilbert-Schmidt (HS) volumes of the $N \times N$ (off-diagonal complex-valued) density matrices \cite[eq. (4.5)]{zyczkowski2003hilbert} 
(cf. \cite[eq. (14.38)]{bengtsson2017geometry}),
\begin{equation} \label{ZSComplex}
V_{HS}^{\mathbb{C}}(N)= \frac{\sqrt{N} (2 \pi )^{\frac{1}{2} (N-1) N} \prod_{i=1}^N \Gamma(i)}{\Gamma(N^2) },   
\end{equation}
and their $\frac{N^2+N-2}{2}$-dimensional real-valued counterparts \cite[eq. (7.7)]{zyczkowski2003hilbert},
\begin{equation} \label{ZSreal}
V_{HS}^\mathbb{R}(N)=  \frac{\sqrt{N} 2^N (2 \pi )^{\frac{1}{4} (N-1) N} \Gamma[(N+1)/2] \prod_{i=1}^N \Gamma(1+i/2)}{\Gamma(N(N+1)/2) \Gamma(1/2) }.     \end{equation}

Further, Andai alternatively employed Lebesgue measure (yielding results equivalent with the use of the normalization factor, $\sqrt{N} 2^{N(N-1)/2}$ \cite[p. 13648]{andai2006volume} to the Hilbert-Schmidt ones), obtaining in the complex case \cite[Thm. 2]{andai2006volume},
\begin{equation} \label{AndaiComplex}
V_{Lebesgue}^\mathbb{C}(N)=\frac{\pi ^{\frac{1}{2} (N-1) N} \Pi_{i=1}^{N-1} i!}{\left(N^2-1\right)!}.
\end{equation}
For  the real case  (we are only immediately interested here in the even dimensions $N=4, 6, 8, 10$), taking $2 l=N$, Andai gave 
\cite[Thm. 1]{andai2006volume},
\begin{equation}  \label{AndaiReal}
V_{Lebesgue}^\mathbb{R}(l)=\frac{2^{-l^2-l} \pi ^{l^2} (2 l)!}{l! \left(2 l^2+l-1\right)!}    \prod_{i=1}^{l-1} (2 i)!.
\end{equation}

To, additionally, obtain the volume formulas with respect to induced measure  \cite{zyczkowski2001induced} \cite[sec. 15.5]{bengtsson2017geometry}, in the two-qubit cases ($N=4$), we must multiply the complex ($\mathbb{C}$) total volume expressions
of {\.Z}yczkowski and Sommers (\ref{ZSComplex}) and of Andai (\ref{AndaiComplex}) for $N=4$ by \cite[eq. (3.7)]{zyczkowski2001induced}
\begin{equation}
 \frac{217945728000 (1)_k (2)_k (3)_k \Gamma (k+4)}{\Gamma (4 (k+4))},   
\end{equation}
where the Pochhammer symbol is indicated.
Similarly, for the qubit-qutrit case ($N=6$), we must multiply by 
\begin{equation}
\frac{86109566386551207747222094479360000000 (1)_k (2)_k (3)_k (4)_k (5)_k \Gamma
   (k+6)}{\Gamma (6 (k+6))}.    
\end{equation}
\subsection{Outline of study}
Our first collection of (qubit-qutrit, qubit-qudit, rebit-retrit, rebit-redit and two-quaterbit) analyses, reported in the immediately following sections (sec.~\ref{qubitqutritsection}-\ref{quaternionicsection}), could be conducted with either  set of the volume formulas referenced above ((\ref{ZSComplex})-(\ref{AndaiReal})). For specificity, we will proceed with the second (Andai/Lebesgue) set (cf. (\ref{MZ1}), (\ref{MZ2})), in investigating higher-dimensional counterparts to the now available extensive collection of very well-supported results for the generalized two-qubit states. In particular, let us note, the conjecture (\ref{qubitqutritconjecture}) that the qubit-qutrit counterpart to the apparent $\frac{8}{33}=\frac{2^3}{3 \cdot 11}$ two-qubit separability probability is $\frac{27}{1000}=(\frac{3}{10})^3 =\frac{3^3}{2^3 \cdot 5^3}$.

In the next secs.~\ref{twoqutritssection} and \ref{Xstates}, we examine separability/PPT-probability issues in the context of the two-qutrits and $X$-states, respectively.

In sec.~\ref{Determinantal}, we study for several scenarios, the division of separability probabilities between  the determinantal inequalities $|\rho^{PT}| > |\rho|$ and $|\rho| > |\rho^{PT} |$, where $\rho^{PT}$ denotes the partial transpose of the density matrix $\rho$. Equidivision occurs with Hilbert-Schmidt measure, but not otherwise.

Further, in secs.~\ref{Bures1}, \ref{Bures2} and \ref{Bures3},  we employ a number of innovative approaches to increase our understanding/estimation of the separability probability with respect to the fundamental Bures (minimal monotone) measure \cite{slater2002priori,slater2005silver}. We also study interpolations of separability probabilities between Hilbert-Schmidt and operator monotone (Bures, $\sqrt{x}$) measures in these sections.

In sec.~\ref{LovasAndaiExtension} and the detailed appendices of Charles Dunkl (sec.~\ref{CFD}) below, we  extend to  induced measure, the ``master Lovas-Andai formula'' for generalized two-qubit states
\begin{equation} \label{MasterFormula}
\tilde{\chi}_{d,0}(\varepsilon) \equiv  \tilde{\chi_d}(\varepsilon)= \frac{\varepsilon ^d \Gamma (d+1)^3 \,
   _3\tilde{F}_2\left(-\frac{d}{2},\frac{d}{2},d;\frac{d}{2}+1,\frac{3
   d}{2}+1;\varepsilon ^2\right)}{\Gamma \left(\frac{d}{2}+1\right)^2},  
\end{equation}
reported in \cite[sec. VIII.A]{slater2017master}.
(Here, $\varepsilon$ is a singular-value ratio, and $d$ the random-matrix Dyson index.)
We will find that the extended formula $\tilde{\chi}_{d,k}(\varepsilon)$ naturally breaks into the sum of two parts, both parts simply reducing to one-half of $\tilde{\chi}_{d,0}(\varepsilon)$ when the induced measure index $k$ equals zero, that is, for the case of Hilbert-Schmidt measure. (The original ``unextended'' formula (\ref{MasterFormula}) applies in that specific case, as well as when the measure employed is that based on the operator monotone function $\sqrt{x}$--as Lovas and Andai showed \cite[sec. 4]{lovas2017invariance}.) 

One of these two parts, which both together sum to  
$\tilde{\chi}_{d,k}(\varepsilon)$, is
\begin{equation} \label{onehalf}
J\left(  \varepsilon\right)     = \frac{(1)_d \varepsilon ^d \Gamma (d+k+1)^2 \,
   _3\tilde{F}_2\left(\frac{d}{2},d,-\frac{d}{2}-k;\frac{d}{2}+1,\frac{3
   d}{2}+k+1;\varepsilon ^2\right)}{d \Gamma \left(\frac{d}{2}\right) \Gamma
   \left(\frac{d}{2}+k+1\right)}.   
\end{equation}
The other complementary part of the extended  master formula takes the form
\begin{align} \label{otherhalf}
I\left(  \varepsilon\right)    & =\frac{2\Gamma\left(  1+d+k\right)
^{2}\Gamma\left(  \frac{d}{2}\right)  \Gamma\left(  d\right)  k!\varepsilon
^{d}}{d\Gamma\left(  \frac{d}{2}\right)  ^{3}\Gamma\left(  1+\frac{d}%
{2}+k\right)  \Gamma\left(  1+k\right)  \Gamma\left(  \frac{3d}{2}+1+k\right)
}\\
& \times\sum_{j=0}^{k}\frac{\left(  d\right)  _{k-j}}{\left(  \frac{d}%
{2}+1\right)  _{k-j}}\left(  1-\varepsilon^{2}\right)  ^{k-j}~_{3}F_{2}\left(
%
\genfrac{}{}{0pt}{}{-d/2-j,d/2,d+k-j}{1+d/2+k-j,1+k+3d/2}%
;\varepsilon^{2}\right)  .
\end{align}
Then, to obtain a generalized ($d$) two-qubit separability with respect to an induced measure corresponding to $k$, one implements the formula
\begin{equation} \label{interpolationGeneral}
\mathcal{P}_{sep/PPT}(d,k)= \frac{\int_{-1}^1  \int_{-1}^x\tilde{\chi}_{d,k} (\sqrt{\frac{1-x}{1+x}}  \sqrt{\frac{1+y}{1-y}})(1-x^2)^{d+k} (1-y^2)^{d+k} (x-y)^d \mbox{d} y \mbox{d} x}{\int_{-1}^1  \int_{-1}^x(1-x^2)^{d+k} (1-y^2)^{d+k} (x-y)^d \mbox{d} y \mbox{d} x}.
\end{equation}
For example, $d=2, k= 1$, with $\tilde{\chi}_{2,1}(\varepsilon)=\frac{1}{4} \varepsilon ^2 \left(3-\varepsilon ^2\right)^2$, yields $\frac{61}{143} =\frac{61}{11 \cdot 13} \approx 0.426573$, as previously reported \cite[eq. (2)]{slater2015formulas} \cite[eq. (4)]{slater2016formulas}. If we replace the $d+k$ terms in (\ref{interpolationGeneral}) by $-\frac{d}{4}+k$, we obtain the separability/PPT-probability function in the operator monotone function $\sqrt{x}$ case.

In our concluding remarks (sec.~\ref{Concluding}), we discuss, among other things, the establshed relevance of {\it Casimir invariants} \cite{slater2016invariance,gerdt20116,byrd2003characterization} to the study of separability probabilities in sec.~\ref{Casimir}.

\section{Qubit-qutrit analyses} \label{qubitqutritsection}
For the two-qubit ($N=4$) case, using (\ref{AndaiComplex}), we have for the 15-dimensional volume of two-qubit states,
\begin{equation} \label{CN4one}
V_{Lebesgue}^\mathbb{C}(4)    = \frac{\pi ^6}{108972864000} =\frac{\pi^6}{2^9 \cdot 3^5 \cdot 5^3 \cdot 7^2 \cdot 11 \cdot 13}.
\end{equation}
Multiplying this by the associated well-supported separability probability $\frac{8}{33}$, we have 
\begin{equation} \label{CN4two}
V_{Sep/Lebesgue}^\mathbb{C}(4)    = \frac{\pi ^6}{449513064000} =\frac{\pi^6}{2^6 \cdot 3^6 \cdot 5^3 \cdot 7^2 \cdot 11^2 \cdot 13}.
\end{equation}
So, we see that the same primes (but to different powers) occur in  the denominators of both volume formulas, while the two numerators remain the same.

Let us now see if we can find analogous behavior in the bipartite ($2 \times 3$) qubit-qutrit ($N=6$) case. On the basis of 2,900,000,000 randomly-generated 
qubit-qutrit density matrices \cite[sec. 4]{al2010random},\cite{zyczkowski2011generating}, we obtained an estimate (with 78,293,301
separable density matrices found) for an associated separability probability of 0.026997690.
(We incorporate the results for one hundred million density matrices reported in  \cite[sec. II]{slater2016invariance}.
Milz and Strunz give a confidence interval of $0.02700 \pm 0.00016$ for this probability \cite[eq. (33)]{milz2014volumes}. A [narrower] $95\%$ confidence interval
based  on our just indicated calculation is $[0.0269918, 0.0270036]$. In the decade-old 2007 paper 
\cite[sec.  10.2]{slater2007dyson}, where the $\frac{8}{33}$ two-qubit conjecture was first formulated, we had advanced a hypothesis of $\frac{32}{1199} =\frac{2^5}{11 \cdot 109} \approx 0.0266889$--subsequently rejected as lying outside the confidence interval reported in \cite[sec.II]{slater2016invariance}. An effort to extend the Lovas-Andai form of analysis
\cite{lovas2017invariance} to the qubit-qutrit and rebit-retrit states has been reported in \cite[App. B]{slater2017master}--but, it now seems, that the separability probabilities reported there were subject to some small, yet not explained,  systematic error.) 

We have for the 35-dimensional volume of qubit-{\it qutrit} states,
\begin{equation}
V_{Lebesgue}^\mathbb{C}(6)    =     \frac{\pi ^{15}}{298991549953302804677854494720000000} =
\end{equation}
\begin{displaymath}
\frac{\pi^{15}}{2^{24} \cdot 3^{12} \cdot 5^7 \cdot 7^5 \cdot 11^3 \cdot 13^2 \cdot 17^2 \cdot 19 \cdot 23 \cdot 29 \cdot 31}.
\end{displaymath}
Now, we have found that, for  a separability probability of 
\begin{equation} \label{qubitqutritconjecture}
\frac{27}{1000} =\frac{3^3}{2^3 \cdot 5^3} =(\frac{3}{10})^3= 0.027,
\end{equation}
we would have the corresponding volume of separable states,
\begin{equation}
V_{Sep/Lebesgue}^\mathbb{C}(6)    =    \frac{\pi ^{15}}{298991549953302804677854494720000000} =
\end{equation}
\begin{displaymath}
\frac{\pi^{15}}{2^{27} \cdot 3^{9} \cdot 5^{10} \cdot 7^5 \cdot 11^3 \cdot 13^2 \cdot 17^2 \cdot 19 \cdot 23 \cdot 29 \cdot 31}.
\end{displaymath}
So, we see that only the powers of 2, 3 and 5 are modified,  closely following the pattern observed ((\ref{CN4one})-(\ref{CN4two})) in the $N=4$ scenario.

A  point to note here is that in the $4 \times 4$ density matrix setting, the positivity of  the determinant of the partial transpose is sufficient for separability to hold \cite{augusiak2008universal}, but not so in the $6 \times 6$ setting. (The partial transpose for an entangled state might  have {\it two} negative eigenvalues \cite{johnston2013non}--but not, we note, in the corresponding $X$-states scenario \cite[App. A]{mendoncca2017maximally}.) The possibility of a pair of  negative eigenvalues renders it less directly useful to employ determinantal moments of density matrices and of their partial transposes to approximate underlying separability probability distributions, as was importantly done in 
\cite{slater2012moment,slater2013concise}, using ``moment-based density approximants'' \cite{provost2005moment}, based on Legendre polynomials.

In App.~\ref{Qubitqutritinduced}, we report parallel qubit-qutrit analyses employing, rather than the Hilbert-Schmidt measure ($k=0$), induced measures ($k \neq 0$). However, at this stage, we do not advance specific conjectures as to their corresponding separability probabilities.

\section{Qubit-qudit  analyses}
\subsection{$2 \times 4$ case}
In \cite[sec. III.B]{slater2016invariance}, we reported a PPT (positive partial transpose) probability, for the $8  \times 8$ density matrices (viewed as $2 \times 4$ systems) of
0.0012923558, based on 348,500,000 random realizations \cite{al2010random}, 450,386 of them having PPT's.  The associated $95\%$ confidence interval is $[0.0128863, 0.0129609]$. (Milz and Strunz did report  an estimate of 0.0013 \cite[Fig. 5]{milz2014volumes}, but gave no 
associated confidence interval or sample size.)

Let us interestingly note that the numerator of the ($2 \times 2$) two-qubit separability probability $\frac{8}{33}$ is $2^3$, and of our ($2 \times 3$) qubit-qutrit conjecture, $\frac{27}{1000}$, it is $3^3$. So, we might speculate that in this $2 \times 4$ setting, the numerator of the PPT-probability would be $4^3 = 64$. Proceeding as in sec.~\ref{qubitqutritsection}, using the Andai Lebesgue volume formula (\ref{AndaiComplex}), with $N=8$,  we did find a 
candidate PPT-probability (but with a numerator of $4^2$) of  $\frac{16}{12375} =\frac{4^2}{3^2 \cdot 5^3 \cdot 11} \approx 0.001292929$.

It would be of interest to try to examine the issue of what proportion of the  $2 \times 4$ PPT-states are, in fact, separable (cf. \cite[sec. IV]{zyczkowski1999volume}), as opposed to bound entangled, using the methodologies recently presented in \cite{qian2018state,Li2018}.
\subsection{$2 \times 5$ case}
We generated 621,000,000 $10 \times 10$ random such density matrices. Of these, 16,205 had a PPT, giving us as
estimated PPT-probability of 0.0000260950. A possible exact value, in line with the noted numerator phenomenon, might be
$\frac{125}{4790016} = \frac{5^3}{2^8 \cdot 3^5 \cdot 5  \cdot 7 \cdot 11} \approx 0.0000260959$.

In a supplementary analysis, for thirty-six  million $10 \times 10$ density matrices, again randomly generated with respect to Hilbert-Schmidt measure, we  found 950 to have PPT's. Among these,  {\it none} passed the further test for {\it separability from spectrum} \cite[Thm. 1]{johnston2013separability}. That is, for 
none, in this 10-dimensional setting, did the condition hold that $\lambda_1< \lambda_9 +2 \sqrt{\lambda_8 \lambda_{10}}$, where the $\lambda$'s are the ten ordered eigenvalues of 
the density matrices, with $\lambda_1$ being the greatest (cf. \cite[App. A]{slater2017master}).

\section{Rebit-retrit analysis} \label{Rebitretritanalysis}
For the two-rebit ($l=2, N=4$) case, we have for the 9-dimensional volume of two-rebit states,
\begin{equation} \label{RN4one}
V_{Lebesgue}^\mathbb{R}(2)    = \frac{\pi ^4}{967680} =\frac{\pi^4}{2^{10} \cdot 3^3 \cdot 5^ \cdot 7}.
\end{equation}
Multiplying this by the established (by Lovas and Andai \cite[Cor. 2]{lovas2017invariance})  separability probability $\frac{29}{64}$, we find
\begin{equation} \label{RN4two}
V_{Sep/Lebesgue}^\mathbb{R}(2)    =\frac{29 \pi ^4}{61931520} =\frac{29 \pi^4}{2^{16} \cdot 3^3 \cdot 5 \cdot 7}.
\end{equation}
So, we see that only the power of 2 is modified, and the exponents of 3, 5 and 7 in the denominators are unchanged.

Let us now see if we can find analogous simple behavior in the rebit-retrit ($l=3, N=6)$ case. On the basis of 3,530,000,000 randomly-generated 
rebit-retrit density matrices \cite[sec. 4]{al2010random}, with respect to Hilbert-Schmidt measure, we obtained an estimate (with 462,704,503
separable density matrices found) for an associated separability probability of 0.1310777629. The associated 
$95\%$ confidence interval is $[0.131067, 0.131089]$.

We have
for the total (20-dimensional) volume of both separable and entangled rebit-retrit states,
\begin{equation}
V_{Lebesgue}^\mathbb{R}(3)    =    \frac{\pi ^9}{1730063650258944000} =\frac{\pi^{9}}{2^{23} \cdot 3^{6} \cdot 5^3 \cdot 7^2 \cdot 11 \cdot 13 \cdot 17 \cdot 19}.
\end{equation}
Then we found that, assuming a very closely fitting separability probability of 
\begin{equation}
\frac{860}{6561} =\frac{2^2 \cdot 5 \cdot 43}{3^8} \approx 0.1310775796,
\end{equation}
we would have
\begin{equation}
V_{Sep/Lebesgue}^\mathbb{R}(3)    = \frac{859 \pi ^9}{11338145138337015398400} =
\end{equation}
\begin{displaymath}
\frac{859 \pi^{9}}{2^{38} \cdot 3^{6} \cdot 5^{2} \cdot 7^2 \cdot 11}.
\end{displaymath}
So, we see that only the powers of 2 and now of 5 in the denominator are again modified.

We note, in the case of $\frac{860}{6561}$, a possible parallism with the conjectured numerators in the qubit-qudit $2 \times m$ cases being powers of $m$, while now in the real cases, the denominators would be.

Let us further observe that the two-rebit counterpart to the two-qubit induced measure formula (\ref{qubitinduced}) 
is \cite[eq. (4)]{slater2015formulas} \cite[eq. (6)]{slater2016formulas},
\begin{equation} \label{rebitinduced}
 P^{2-rebits}_k=1-\frac{4^{k+1} (8 k+15) \Gamma{(k+2)} \Gamma{(2 k+\frac{9}{2}})}{\sqrt{\pi} \Gamma{(3 k+7)}}.
\end{equation}
For the Hilbert-Schmidt $k=0$ case, we obtain the formally demonstrated result, $\frac{29}{64}$.
\section{Rebit-redit analyses}
\subsection{$2 \times 4$ case}
We generated  490,000,000 $8 \times 8$ random density matrices with respect to Hilbert-Schmidt ($k=0$) measure. Of these, 12,022,129 had a PPT, giving us as
estimated PPT-probability of 0.02453496. A good fit is provided by $\frac{201}{8192} = \frac{3 \cdot 67}{2^{13}} \approx 0.0245361$. We note, in light of 
our previous analyses, that the denominator $2^{13}$ is obviously also expressible as $4^{6+\frac{1}{2}}$.
\subsection{$2 \times 5$ case}
We generated  620,000,000 $10 \times 10$ random density matrices with respect to Hilbert-Schmidt ($k=0$) measure. Of these, 1,844,813 had a PPT, giving us as
estimated PPT-probability of 0.002975505. A well-fitting candidate PPT-probability is $\frac{29058}{9765625}= \frac{2 \cdot 3 \cdot 29 \cdot 167}{5^{10}} \approx 0.00297554$.
\section{Quaternionic formulas} \label{quaternionicsection}
Let us also note that in \cite[Thm. 3]{andai2006volume}, Andai presented the quaternionic ($\mathbb{H}$) counterpart,
\begin{equation}
 V_{Lebesgue}^\mathbb{H}(N)=   \frac{\pi ^{N^2-N} (2 N-2)!}{\left(2 N^2-N-1\right)!} \prod_{i=1}^{N-2}(2 i)!,
\end{equation}
of the complex ($\mathbb{C}$) and real ($\mathbb{R}$) volume formulas ((\ref{AndaiComplex}), (\ref{AndaiReal})) given above.
We, then have for the 27-dimensional volume of the two-quaterbit states,
\begin{equation}
 V_{Lebesgue}^\mathbb{H}(4)=   \frac{\pi ^{12}}{315071454005160652800000} =
\end{equation}
\begin{displaymath}
\frac{\pi^{12}}{2^{15} \cdot 3^{10} \cdot 5^5 \cdot 7^3 \cdot 11^2 \cdot 13^2 \cdot 17 \cdot 19 \cdot 23}.
\end{displaymath}
Multiplying by the well-supported separability/PPT-probability value (cf. \cite{hildebrand2008semidefinite}) of $\frac{26}{323}$, 
we find
\begin{equation}
 V_{Sep/Lebesgue}^\mathbb{H}(4)=  \frac{\pi ^{12}}{3914156909371803494400000} =
\end{equation}
\begin{displaymath}
\frac{\pi^{12}}{2^{14} \cdot 3^{10} \cdot 5^5 \cdot 7^3 \cdot 11^2 \cdot 13^2 \cdot 17 \cdot 19 \cdot 23}.
\end{displaymath}
We would like to extend our earlier analyses above to the (50-dimensional) ``quaterbit-quatertrit'' setting. But it is 
clearly a challenging problem
to suitably generate sufficient numbers of random $6 \times 6$ density matrices of such a nature (cf. \cite[App. C]{slater2017master}  of C. Dunkl), in order to obtain the needed probability estimates to attempt to closely fit.

We further note that the two-quaterbit counterpart to the two-qubit  and two-rebit induced measure formulas (\ref{qubitinduced}) and (\ref{rebitinduced}), 
is \cite[eq. (3)]{slater2015formulas} \cite[eq. (5)]{slater2016formulas},
\begin{equation} \label{quaterbitinduced}
 P^{2-quaterbits}_k=1-\frac{4^{k+6} (k(k(2 k(k+21)+355)+1452)+2430)\Gamma{(k+\frac{13}{2})} \Gamma{(2 k+15)}}{3 \sqrt{\pi} \Gamma{(3 k+22)}}.
\end{equation}
\section{Two-qutrits} \label{twoqutritssection}
In \cite[sec. III.A]{slater2016invariance}, we reported an estimated Hilbert-Schmidt PPT-probability  of 0.00010218 for the two-{\it qutrit} states \cite{baumgartner2006state}, based on one hundred million randomly generated density matrices. Following the framework employed above, we have made some limited efforts to suggest a possible corresponding exact probability. It is by no means clear, however, if one can hope to extend ($2 \times m$) qubit-based results to a fully qutrit setting. (In any case, we did find that the rational value $\frac{323}{3161088}= \frac{17 \cdot 19}{2^{10} \cdot 3^2 \cdot 7^3} \approx 0.00010218$ provides an exceptional fit.) It would be of interest to try to examine the issue of what proportion of the two-qutrit PPT-states are, in fact, separable (cf. \cite{zyczkowski1998volume}) using the methodologies recently presented in \cite{qian2018state,Li2018}.
\section{$X$-states} \label{Xstates}
We have found that 
the Hilbert-Schmidt separability/PPT-probabilities for both the ($6 \times 6$) rebit-retrit and ($9 \times 9$) two-retrit 
$X$-states to be, somewhat remarkably, equal to that previously reported \cite[p. 3]{dunkl2015separability} for the lower-dimensional ($4 \times 4$) two-rebit $X$-states, that is, $\frac{16}{3 \pi^2} \approx 0.54038$. (The HS two-qubit $X$-states 
separability probability has previously
been shown to equal $\frac{2}{5} = 0.4$ \cite[eq. (22)]{milz2014volumes} \cite[p. 3]{dunkl2015separability}. In \cite[App. B]{slater2017master}, we noted that Dunkl had concluded that the same separability probability did hold for the qubit-qutrit $X$-states.)  

We have also found that the equality between two-rebit and rebit-retrit $X$-states separability probabilities continues to hold when the Hilbert-Schmidt measure (the case $k=0$) is 
generalized to the class of induced measures \cite{zyczkowski2001induced,bengtsson2017geometry}. In Fig. 1, we present two equivalent formulas that yield these induced measure two-rebit, rebit-retrit separability probabilities.
\begin{figure}
\includegraphics[page=1,scale=0.9]{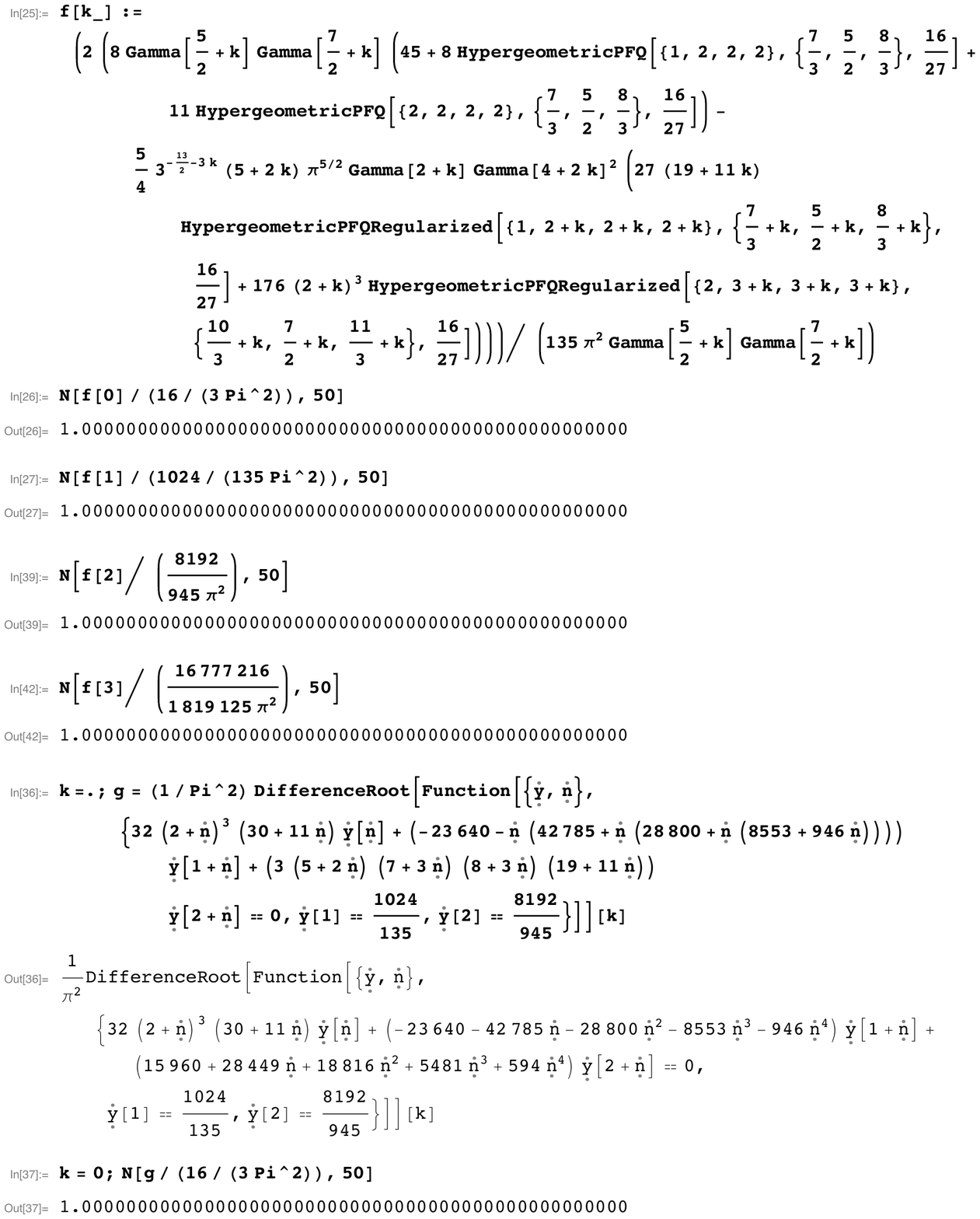}
\label{fig:Induced}
\end{figure}
\begin{figure}
\includegraphics[page=2,scale=0.9]{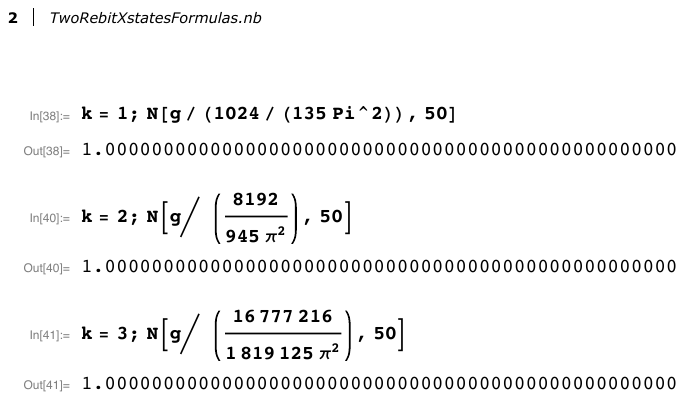}
\caption{Two-rebit and rebit-retrit $X$-states induced separability probability formulas}
\end{figure}
\section{Determinantal equipartition of Hilbert-Schmidt separability probabilities} \label{Determinantal}
In \cite{slater2016formulas}, a formula (eq. (62) there),
\begin{align*} \label{Hyper1}
Q\left(  k,\alpha\right)   &  =\frac{1}{2}-\frac{\alpha\left(  20\alpha+8k+11\right)
\Gamma\left(  5\alpha+2k+2\right)  \Gamma\left(  3\alpha+k+\frac{3}{2}\right)
\Gamma\left(  2\alpha+k+\frac{3}{2}\right)  }{4\sqrt{\pi}\Gamma\left(
5\alpha+2k+\frac{7}{2}\right)  \Gamma\left(  \alpha+k+2\right)  \Gamma\left(
4\alpha+k+2\right)  }\\
&  \times~_{6}F_{5}\left(
\genfrac{}{}{0pt}{}{1,\frac{5}{2}\alpha+k+1,\frac{5}{2}\alpha+k+\frac{3}{2}%
,2\alpha+k+\frac{3}{2},3\alpha+k+\frac{3}{2},\frac{5}{2}\alpha+k+\frac{19}{8}%
}{\alpha+k+2,4\alpha+k+2,\frac{5}{2}\alpha+k+\frac{7}{4},\frac{5}{2}\alpha+k+\frac{9}{4},\frac
{5}{2}\alpha+k+\frac{11}{8}}%
;1\right)
\end{align*}
was given for that part of the {\it total} induced-measure separability probability, $P(k,\alpha)$, for 
generalized (real [$\alpha=\frac{1}{2}$], complex [$\alpha=1$], quaternionic [$\alpha=2$],...) two-qubit states for which the determinantal inequality $|\rho^{PT}| > |\rho|$ holds. For the 
Hilbert-Schmidt case ($k=0$) the formula yielded $\frac{P(0,\alpha)}{2}$. (In \cite{szarek2006structure}--making use of Archimedes’ formula for the 
volume of a D-dimensional pyramid of unit height, and of ``pyramid-decomposability''--it was shown that the 
Hilbert-Schmidt separability 
probability of the minimially degenerate states is, likewise, {\it one-half} of that of the nondegenerate states.) Our simulations appear to indicate that this
equal division of separability probabilities continues to hold in the rebit-retrit and qubit-qutrit cases. Based on 96,350,607 separable rebit-retrit cases, the estimated proportion for which  $|\rho^{PT}| > |\rho|$ held was 0.499987, and based on 9,450,652 separable qubit-qutrit cases, the companion estimated proportion  was 0.500033. 

However, in the non-Hilbert-Schmidt analysis in sec.~\ref{k1K7}, pertaining to the qubit-qutrit states with induced measure parameters $k=1$, $K=7$, we
found that the determinantal inequality $|\rho^{PT}| > |\rho|$ held in only $31.17\%$ (and not $50\%$) of the cases. Further, in sec.~\ref{k2K8}, pertaining to the qubit-qutrit states with induced measure parameters $k=2$, $K=8$, the corresponding percentage was $22.63\%$.
\section{Lovas-Andai-type formula for operator monotone measures}
It is of clear interest to extend the forms of analysis above to measures of interest other than the Hilbert-Schmidt
(flat/Euclidean/Frobenius) one, in particular perhaps, the Bures (minimal monotone) one (cf. \cite{slater2000exact,vsafranek2017discontinuities}). In these regards, in \cite[sec. VII.C]{slater2017master}, we recently reported, building upon analyses of Lovas and Andai \cite[sec. 4]{lovas2017invariance}, a two-qubit separability probability equal to $1 -\frac{256}{27 \pi^2} =1- \frac{2^8}{3^3 \pi^2} \approx 0.0393251$. This was based on another 
(of the infinite family of) operator monotone functions, namely 
$\sqrt{x}$. (Let us note that the complementary ``entanglement probability'' is simply $\frac{256}{27 \pi^2} \approx 0.960675$. There appears to be no intrinsic reason
to prefer/privilege one of these two forms of probability to the other [cf. \cite{dunkl2015separability}].  We observe that the upper-limit-of-integration variable  denoted $K_s:=\frac{(s+1)^{s+1}}{s^s}$, equalling $\frac{256}{27} =\frac{4^4}{3^3}$, for $s=3$, is frequently employed in the Penson-{\.Z}yczkowski paper, ``Product of Ginibre matrices: Fuss-Catalan and Raney 
distributions'' \cite[eqs. (2),  (3)]{penson2011product}.)  
\subsection{Operator monotone measure $\sqrt{x}$}
Within the Lovas-Andai framework, employing the previously reported two-qubit ``separability function" $\tilde{\chi}_2(\varepsilon)=\frac{1}{3} \varepsilon^2 (4 -\varepsilon^2)$ 
\cite[eq. (42)]{slater2017master}, we can interpolate between the computation for the noted ($\sqrt{x}$) operator monotone separability probability of $1- \frac{256}{27 \pi^2}$ ($\eta=-\frac{1}{2}$) 
and the computation for  the Hilbert-Schmidt counterpart of $\frac{8}{33}$ ($\eta=2$).
This is accomplishable using the formula (Fig.~\ref{fig:InterpolatedHSmonotone}),
\begin{equation} \label{interpolation}
u(\eta) =   \frac{\int_{-1}^1  \int_{-1}^x\tilde{\chi}_2 (\sqrt{\frac{1-x}{1+x}}  \sqrt{\frac{1+y}{1-y}})(1-x^2)^\eta (1-y^2)^\eta (x-y)^2 \mbox{d} y \mbox{d} x}{\int_{-1}^1  \int_{-1}^x(1-x^2)^\eta (1-y^2)^\eta (x-y)^2 \mbox{d} y \mbox{d} x}= 
\end{equation}
\begin{displaymath}
-\frac{-3 \eta  (\eta +4) ((\eta -6) \eta -15)+\frac{16^{2 \eta +3} ((\eta -10) \eta -5)
   \Gamma \left(\eta +\frac{3}{2}\right) \Gamma \left(\eta +\frac{5}{2}\right)^3}{\pi ^2
   (2 \eta +3) \Gamma (4 \eta +5)}+60}{3 (\eta -1)^2 \eta ^2},  
\end{displaymath}
where $u(-1)=0$ and $u(1)=\frac{41471}{105}-40 \pi ^2 \approx 0.177729$. 
(It is not now clear if any particularly meaningful measure-theoretic/quantum-information-theoretic interpretation can be given to these interpolated values.)
\begin{figure}
    \centering
    \includegraphics{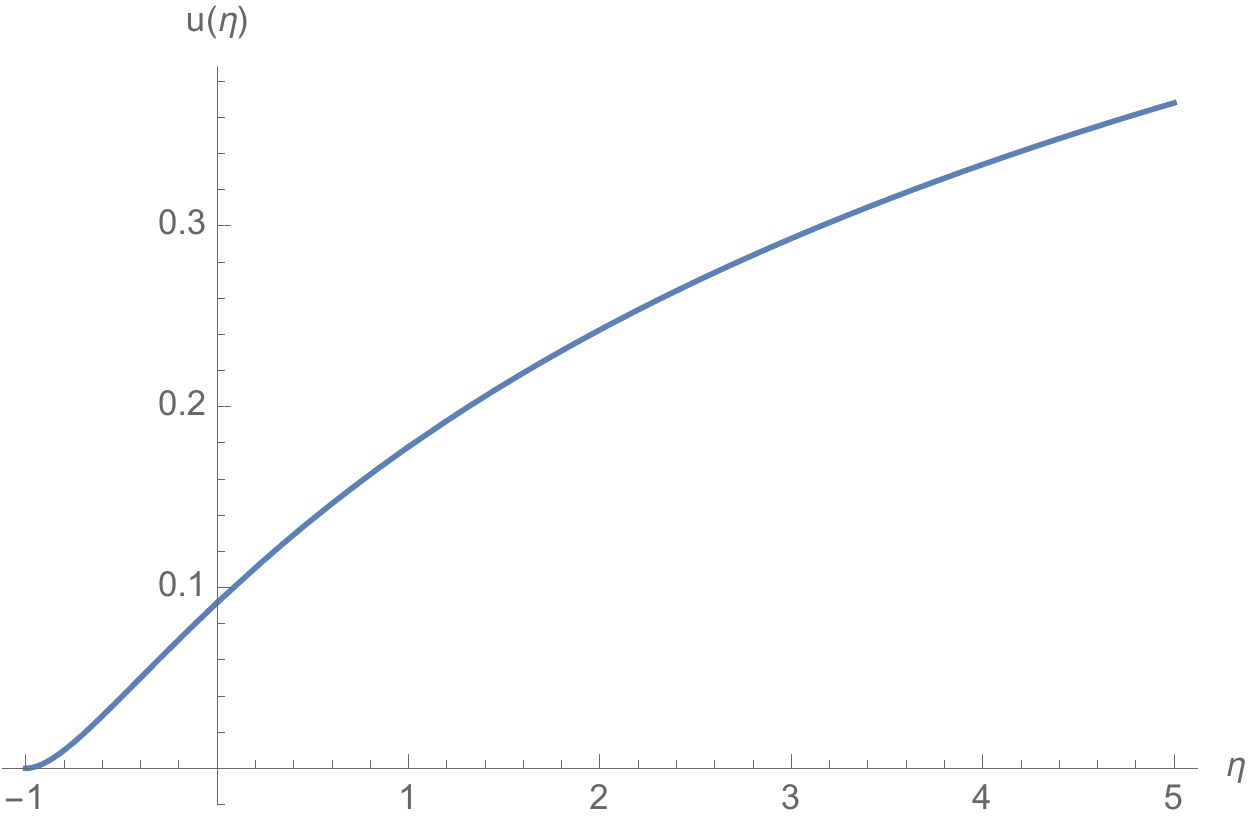}
    \caption{Two-qubit separability probability function $u(\eta)$, given by eq.  (\ref{interpolation}), interpolating between the  ($x \rightarrow \sqrt{x}$) operator monotone result ($\eta=-\frac{1}{2}$) 
    of $1- \frac{256}{27 \pi^2}$ and the Hilbert-Schmidt result ($\eta=2$) of 
    $\frac{8}{33}$. The average of the two is attained at $\eta=0.53544108.$}
    \label{fig:InterpolatedHSmonotone}
\end{figure}
``We argue that from the separability probability point of view, the main difference between the Hilbert-Schmidt measure and the volume form
generated by the operator monotone function $x \rightarrow \sqrt{x}$ is a special distribution on the unit ball in operator norm of 
$2 \times 2$ matrices, more precisely in the Hilbert-Schmidt case one faces a uniform distribution on the whole unit ball and for 
monotone volume forms one obtains uniform distribution on the surface of the unit ball'' \cite[p. 2]{lovas2017invariance}. 
\subsection{Bures measure} \label{Bures1}
Let us, in these regards, also interestingly note that Osipov, Sommers and {\.Z}yzckowski remarked that one can readily interpolate between the generation of random
density matrices with respect to Hilbert-Schmidt and Bures measures \cite[eq. (24)]{al2010random} (cf. \cite[eq. (33)]{borot2012purity}). Their formula took the form
\begin{equation}  \label{JointBuresHSformula}
\rho_x= \frac{(y \mathbb{I} +x U) AA^{\dagger}(y \mathbb{I}+x U^{\dagger})}{\mbox{Tr} (y \mathbb{I}+x U) A A^{\dagger} (y \mathbb{I}   +x U^{\dagger})},
\end{equation}
where $y=1-x$ and $x=0$ yields the Hilbert-Schmidt case, and $x=\frac{1}{2}$, the Bures instance. Let us now perform the change-of-variable
$\eta=2 -5 x$ in eq. (\ref{interpolation}). This allows us to present Fig.~\ref{fig:TwoInterpolations}, in which we simultaneously show 
the Hilbert-Schmidt/Bures interpolation--based on 550,000 sets of 101 randomly-generated density matrices--together with the Hilbert-Schmidt-operator monotone function $x \rightarrow \sqrt{x}$ interpolation, shown using
the $\eta$ variable in Fig.~\ref{fig:InterpolatedHSmonotone}.
\begin{figure}
    \centering
    \includegraphics{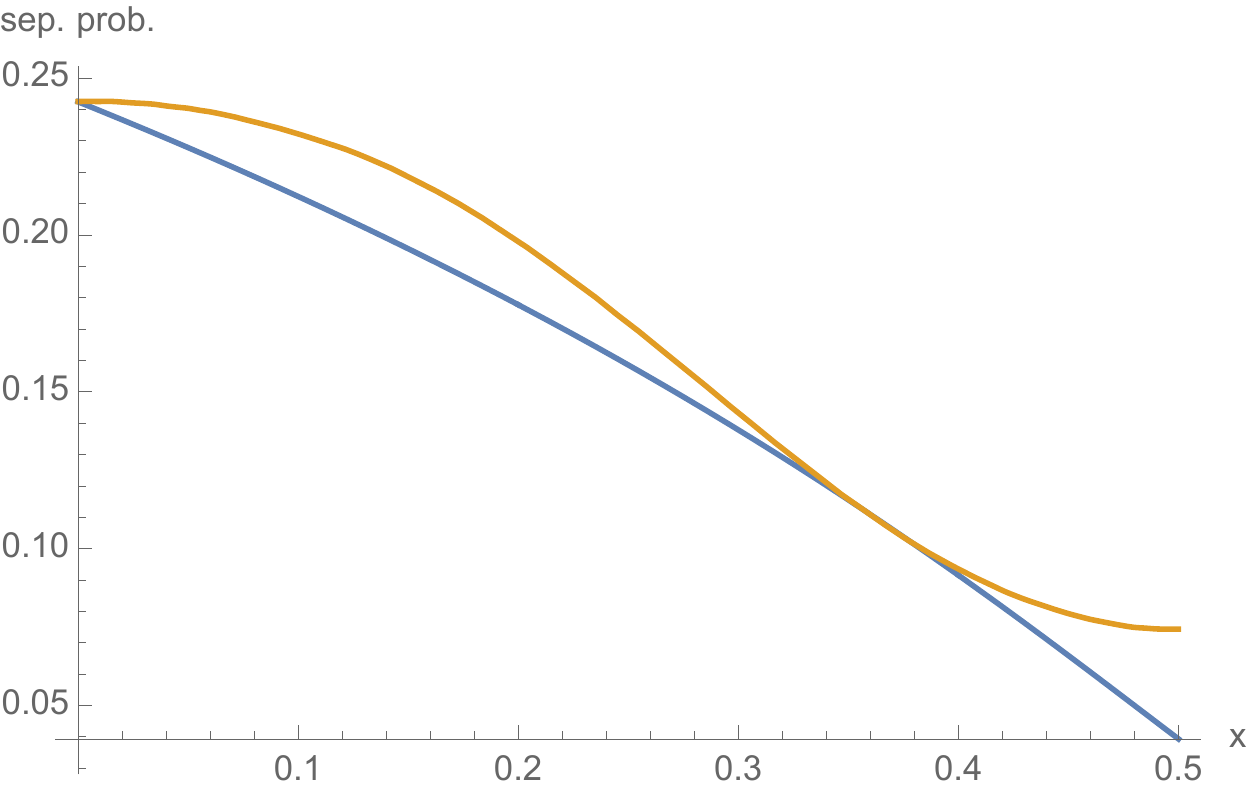}
    \caption{Joint plot of the Hilbert-Schmidt/Bures-based on 550,000 sets of 101 randomly-generated density matrices--and Hilbert-Schmidt/operator monotone $x \rightarrow \sqrt{x}$ interpolations. At $x=\frac{1}{2}$, the former curve has an estimated value of 0.073 and the latter, the 
    lower exact value $1- \frac{256}{27 \pi^2} \approx 0.0393251$. The curves intersect at $x=0.380241$.}
    \label{fig:TwoInterpolations}
\end{figure}
In Fig.~\ref{fig:OneInterpolation}, we more directly interpolate between the HS and Bures two-qubit separability probabilities employing (\ref{JointBuresHSformula}).
\begin{figure}
    \centering
    \includegraphics{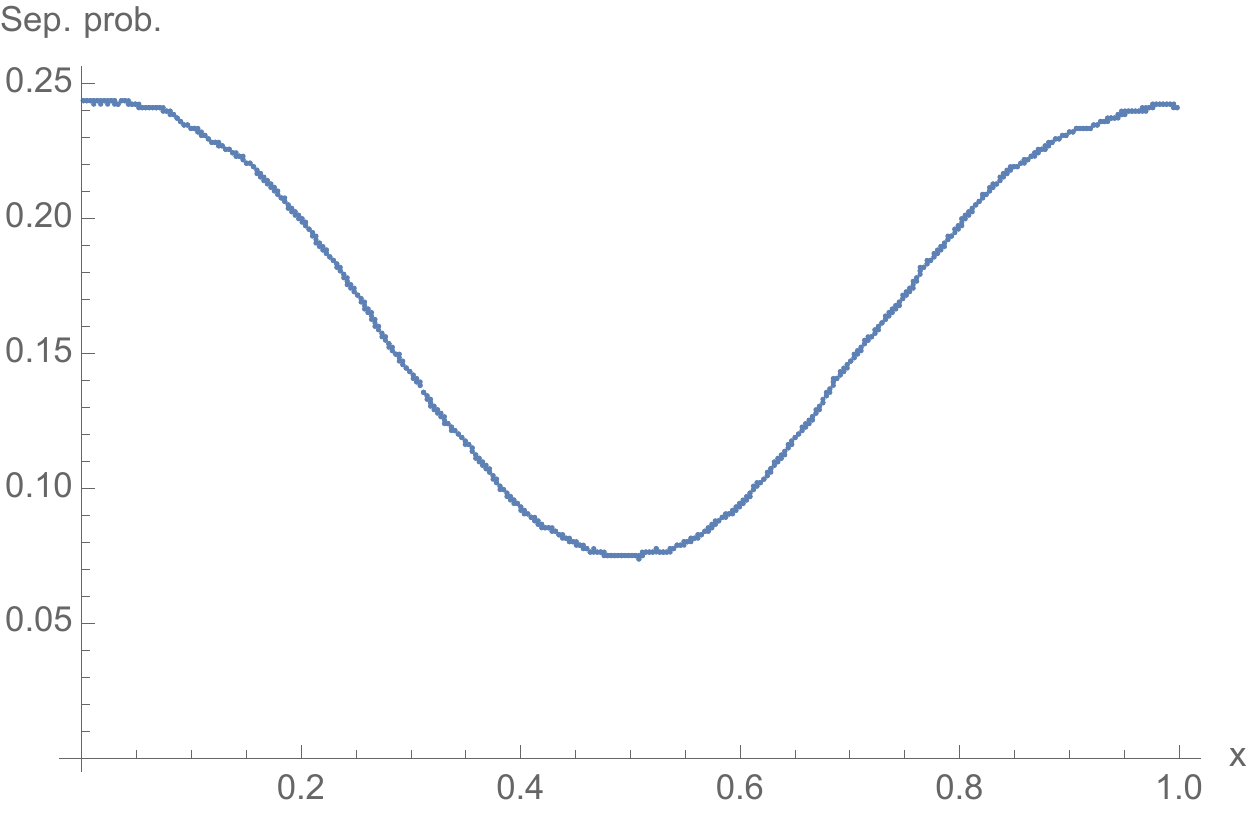}
    \caption{Cosine-like interpolating curve between the Hilbert-Schmidt separability probability of $\frac{8}{33}$ and the (yet not exactly known) Bures separability probability ($x=\frac{1}{2}$) based upon (\ref{JointBuresHSformula}). One million random realizations--with two thousand points in $x \in [0,1]$ per realization--were employed.}
    \label{fig:OneInterpolation}
\end{figure}

Perhaps it is not too unreasonable to anticipate that the Bures two-qubit separability probability 
(associated with the operator monotone function $\frac{1+x}{2}$)--though perhaps not so much the two-rebit case in light of \cite[Thm. 4]{lovas2017invariance}--will also be found to assume a strikingly elegant form. (In 
\cite{slater2002priori}, we had conjectured a value of $\frac{8}{11 \pi^2} \approx .0736881$. But it was later proposed   
in \cite{slater2005silver}, in part motivated by the lower-dimensional {\it exact} results reported in \cite{slater2000exact}, that the value might be $\frac{1680 \sigma_{Ag}}{\pi^8} \approx 0.07334$, where $\sigma_{Ag}= \sqrt{2}-1 \approx 0.414214$ is the ``silver mean''. Both of these studies \cite{slater2002priori,slater2005silver}   were conducted using quasi-Monte Carlo procedures, before the reporting of the  Ginibre-ensemble methodology for generating density matrices, random with respect to the Bures measure \cite{al2010random}. In \cite[Table 1]{slater2005silver} an estimate of 0.0732398 was reported, based on two billion quasi-Monte Carlo [scrambled Tezuka-Faure] points.) 
In 
\cite[sec. VII]{slater2016invariance}, it was noted that ``on the other hand, clear evidence has been provided 
that the apparent $r$-invariance phenomenon revealed by the work of Milz and Strunz,\ldots, does {\it not} continue to hold if one employs, 
rather than Hilbert-Schmidt measure, its Bures (minimal monotone) counterpart''. It would be of interest to examine this issue of $r$-invariance in the context of the induced measures (which, of course, include the Hilbert-Schmidt measure as the special $k=0$ case). 

Another phenomenon apparently restricted to the Hilbert-Schmidt case, is that the separable states are equally divided (in terms of HS measure) between those for which 
$|\rho^{PT}| >|\rho|$, and {\it vice versa}. Based on some 122,000,000 two-qubit density matrices randomly generated with respect 
to Bures measure, of the 8,945,951 separable ones, 5,894,648 of them (that is, 65.8918$\%$) had $|\rho^{PT}| >|\rho|$, clearly distinct from simply one-half of them.
\subsubsection{Hilbert-Schmidt-assisted estimation of two-qubit Bures separability probability} \label{Bures2}
Let us attempt to exploit the Bures-Hilbert-Schmidt interpolation formula  (\ref{JointBuresHSformula}) of Osipov, Sommers and {\.Z}yczkowski, in light of our recently-acquired
high degree of certainty that the Hilbert-Schmidt two-qubit separability probability is $\frac{8}{33}$ 
\cite{slater2017master,lovas2017invariance,milz2014volumes,fei2016numerical,shang2015monte,slater2013concise,slater2012moment,slater2007dyson}. That is, we {\it jointly} estimate the Ginibre matrix $A$ at each
of 4,372,000,000 iterations. (Of these, 1,059,902,370 and 320,546,752 corresponded to separable density matrices, in terms of the Hilbert-Schmidt and Bures measures, respectively.) Then, we apply the average of 1 and the ``correction factor" 1.00002224983, needed to transform the Hilbert-Schmidt estimate to
precisely $\frac{8}{33}$, to the estimated Bures probability. (The average is employed, since a second set of an equal number of 32 [but now independent] standard normal random variates are needed to generate the random matrices with respect to Bures measure.)  Doing so, gives us a very slightly ``corrected'' Bures estimate'' of 0.07331891996, rather than 
0.0733181043.
(Since the additional randomly-generated $4 \times 4$ unitary matrix $U$ enters into the Bures computation--but not the Hilbert-Schmidt one, since $x=0$ in that case--it does seem plausible that convergence is slower in the Bures case.)

Performing a parallel (but much smaller) computation in the two-{\it rebit} 
case, based on forty million random density matrices (6,286,209 of them being separable), we obtain a corresponding (slightly corrected) Bures separability probability estimate of 0.1571469. (In doing so we take, as required, the Ginibre matrix $A$ in (\ref{JointBuresHSformula}) to be $4 \times 5$ \cite[eqs. (24),(28)]{al2010random}, and not $4 \times 4$ as in the two-qubit calculation.)
\section{Estimation of Bures separability probabilities using  quasirandom sequences} \label{Bures3}
We have begun to explore the question of whether the estimation of the  Bures two-qubit and two-rebit separability probabilities could be accelerated--with faster convergence properties--by, rather than using simply
random generation of normal variates for the Ginibre ensembles, making use of normal variates generated by employing low-discrepancy (quasi-Monte Carlo) sequences \cite{leobacher2014introduction}. In particular, we are employing an ``open-ended'' sequence (based on extensions of the {\it golden ratio}) recently introduced by Martin Roberts in the detailed "blog post", ``The Unreasonable Effectiveness 
of Quasirandom Sequences'' \cite{Roberts}.
 
Roberts notes: ``The solution to the 
$s$-dimensional problem, depends on a special constant $\phi_s$, where $\phi_s$ is the value of smallest, positive real-value of x such that''
\begin{equation}
  x^{s+1}=x+1,
\end{equation}
($s=1$, yielding the golden ratio, and $s=2$, the ``plastic constant'' (https://mathematica.stackexchange.com/questions/143457/how-can-one-generate-an-open-ended-sequence-of-low-discrepancy-points-in-3d). (For notational clarity here, we have changed his parameter $d$ to $s$.)
The  $n$-th terms in the quasirandom (Korobov) sequence take the form
\begin{equation} \label{QR}
(\alpha _0+n \vec{\alpha}) \bmod 1, n = 1, 2, 3, \ldots  
\end{equation}
where we have the $s$-dimensional vector,
\begin{equation} \label{quasirandompoints}
\vec{\alpha} =(\frac{1}{\phi_s},\frac{1}{\phi_s^2},\frac{1}{\phi_s^3},\ldots,\frac{1}{\phi_s^s}).
\end{equation}
The additive constant $\alpha_0$ is typically taken to be 0. ``However, there are some arguments, relating to symmetry, that suggest that $\alpha_0=\frac{1}{2}$
is a better choice,''  Roberts observes. These points, lying in the $s$-dimensional hypercube $[0,1]^s$, can be converted to (quasirandomly distributed) normal variates using the inverse of the cumulative distribution function \cite[Chap. 2]{devroye1986}. 
(Henrik Schumacher developed a specialized algorithm, that we employ, that accelerated the  Mathematica command InverseCDF approximately {\it ten-fold}, as reported at   https://mathematica.stackexchange.com/questions/181099/can-i-use-compile-to-speed-up-inversecdf). We take $s=36$ and 64 to estimate the Bures two-rebit and two-qubit separability probabilities, respectively.

In Fig.~\ref{fig:Roberts2Rebits}, we show the paired estimates of the Bures two-rebit separability probability based on $\alpha_0=0$ and $\alpha_0=\frac{1}{2}$ sequences, and similarly in Fig.~\ref{fig:Roberts2Qubits} for the two-qubit case. We see strong convergence of the paired estimates in both figures. For the (longer) $\alpha_0 =\frac{1}{2}$ sequences, the last recorded estimates are 0.1570962 (based on 3,235,000,000 points) and 0.0733116 (2,280,000,000 points), and for the $\alpha_0= 0$ sequences, the last recorded estimates are 0.1570974 (2,685,000,000 points) and 0.0733143 (2,075,000,000 points).
\begin{figure}
    \centering
    \includegraphics{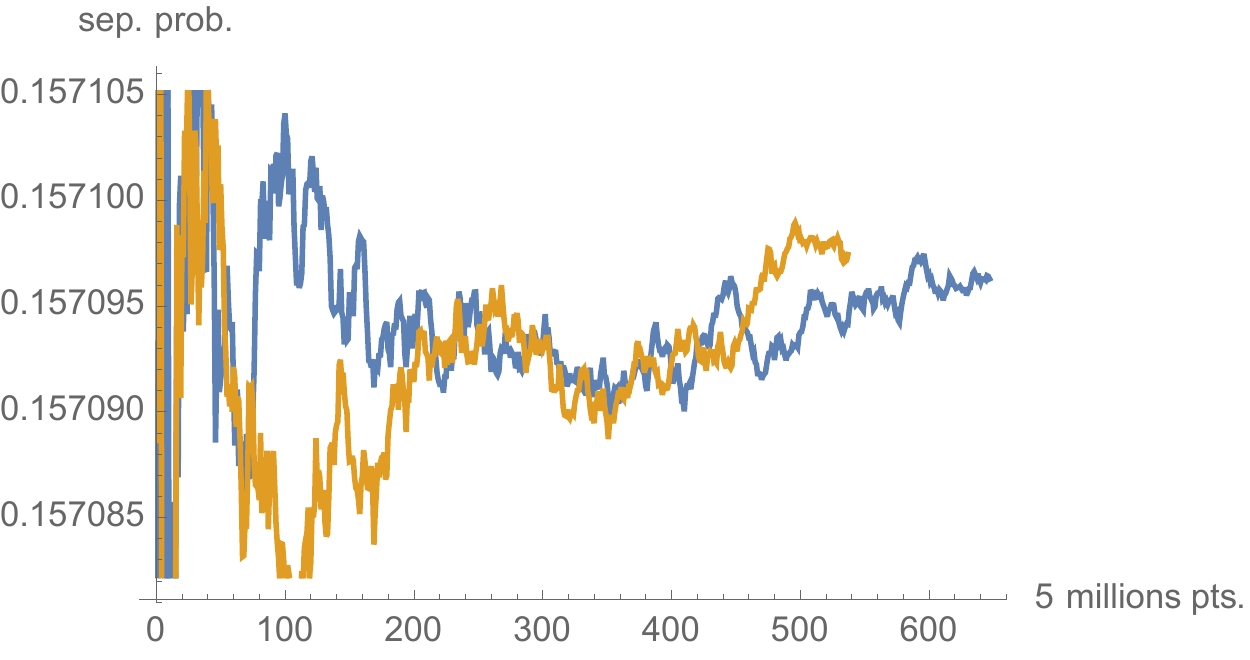}
    \caption{Bures two-rebit separability probability estimates--in units of five million points sampled--based on $\alpha_0=0$ and $\alpha_0=\frac{1}{2}$ quasirandom sequences (\ref{QR}), with the sampling dimension $s=36$. The  $\alpha=\frac{1}{2}$ curve is longer.}
    \label{fig:Roberts2Rebits}
\end{figure}
\begin{figure}
    \centering
    \includegraphics{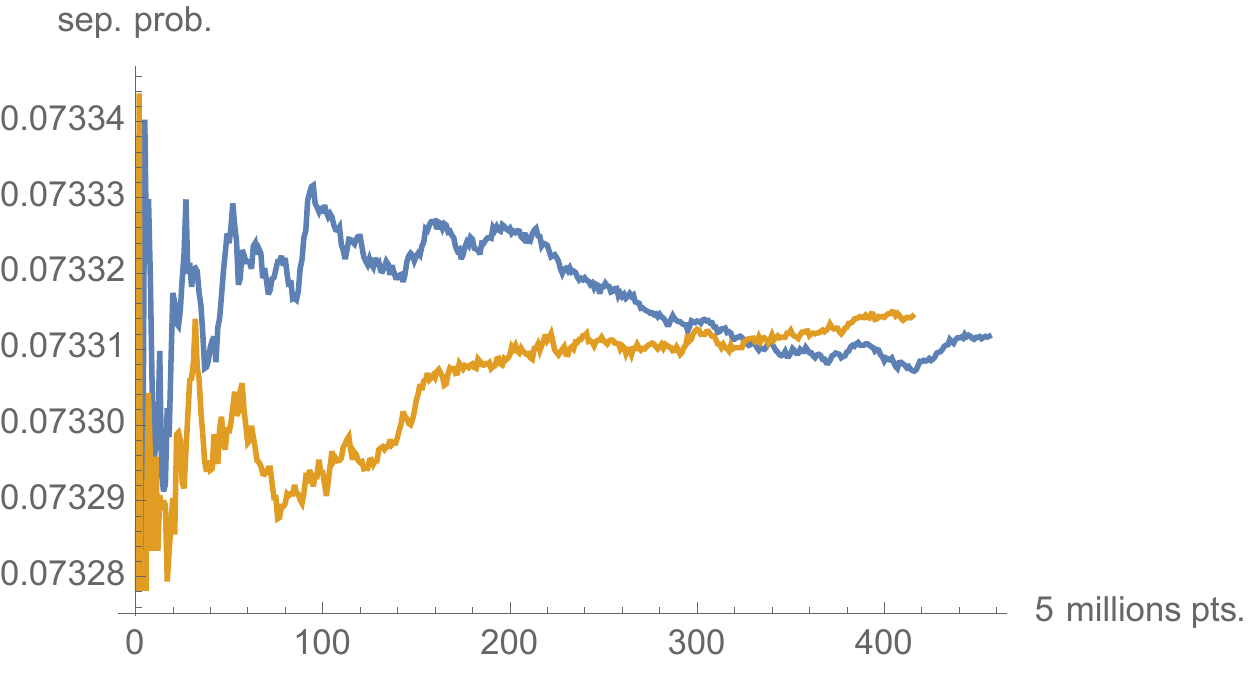}
    \caption{Bures two-qubit separability probability estimates--in units of five million points sampled--based on $\alpha_0=0$ and $\alpha_0=\frac{1}{2}$ quasirandom sequences (\ref{QR}), with the sampling dimension $s=64$. The  $\alpha=\frac{1}{2}$ curve is longer.}
    \label{fig:Roberts2Qubits}
\end{figure}

Our current (significant digits) estimates of the two-rebit and two-qubit Bures separability probabilities, based on these results, are 0.15709 and 0.07331, respectively \cite[Table 1]{khvedelidze2018generation}. We continue with these analyses, aspiring to obtain insights into possible {\it exact} values.

It is of interest to compare/contrast the relative merits of estimation of this pair of Bures separability probabilities in the present 36- and 64-dimensional settings with earlier studies
(largely involving Euler parameterizations of $4 \times 4$ density matrices \cite{tilma2002parametrization}), in which 9- and 15-dimensional integration problems were addressed \cite{slater2005silver,slater2009eigenvalues}. In the higher-dimensional frameworks, the integrands are effectively unity, while not so in the other cases.

``Observe that the Bures volume of the set of mixed states is equal to the volume of an $(N^2-1)$-dimensional hemisphere of radius $R_B=\frac{1}{2}$'' \cite[p. 415]{bengtsson2017geometry}. It is also noted there that $R_B$ times the area-volume ratio asymptotically increases with the dimensionality $D=N^2-1$, which is typical for hemispheres.
\section{Extension of Lovas-Andai formulas to induced measures} \label{LovasAndaiExtension}
Now, let us raise what appears to be a quite interesting research question. That is, can the  Lovas-Andai framework, which has been successfully applied using both Hilbert-Schmidt and operator monotone function $\sqrt{x}$ measures \cite{lovas2017invariance,slater2017master}, be further adopted to the generalization of Hilbert-Schmidt measure to its induced extensions (through the use of the determinantal powers of density matrices in the derivations)? If so, the specific induced separability probabilities reported in 
\cite{slater2015formulas} \cite{slater2016formulas}, including formulas (\ref{qubitinduced}), (\ref{rebitinduced}) and (\ref{quaterbitinduced}) above, could be presumably further verified. We now investigate this topic.
\subsection{Extended Lovas-Andai functions $\tilde{\chi}_{d,k}$ for a few specific cases}
To begin, let us replace $\tilde{\chi}_2(\varepsilon)=\frac{1}{3} \varepsilon^2 (4 -\varepsilon^2)$  in 
the middle expression in the two-qubit separability probability formula (\ref{interpolation}) for $u(\eta)$ by
\begin{equation} \label{d=2k=1}
\tilde{\chi}_{2,1}(\varepsilon)=\frac{1}{4} \varepsilon ^2 \left(3-\varepsilon ^2\right)^2, 
\end{equation}
and set $\eta=3$ (it now being understood, notationally, that $\tilde{\chi}_{2,0}(\varepsilon) \equiv \tilde{\chi}_{2}(\varepsilon)$).

Then, this expression does, in fact, evaluate to the previously-found two-qubit induced $k=1$ value $\frac{61}{143}$ given by formula (\ref{qubitinduced}).
That is,
\begin{equation} \label{interpolation2}
\frac{\int_{-1}^1  \int_{-1}^x\tilde{\chi}_{2,1} (\sqrt{\frac{1-x}{1+x}}  \sqrt{\frac{1+y}{1-y}})(1-x^2)^3 (1-y^2)^3 (x-y)^2 \mbox{d} y \mbox{d} x}{\int_{-1}^1  \int_{-1}^x(1-x^2)^3 (1-y^2)^3 (x-y)^2 \mbox{d} y \mbox{d} x}=  \frac{61}{143}.
\end{equation}
Fig.~\ref{fig:randominducefitk1} shows the residuals from a (clearly close) fit of  $\tilde{\chi}_{2,1}(\varepsilon)$ to an  estimation of it based on
sixty million appropriately generated $4 \times 4$ density matrices.
\begin{figure}
    \centering
    \includegraphics{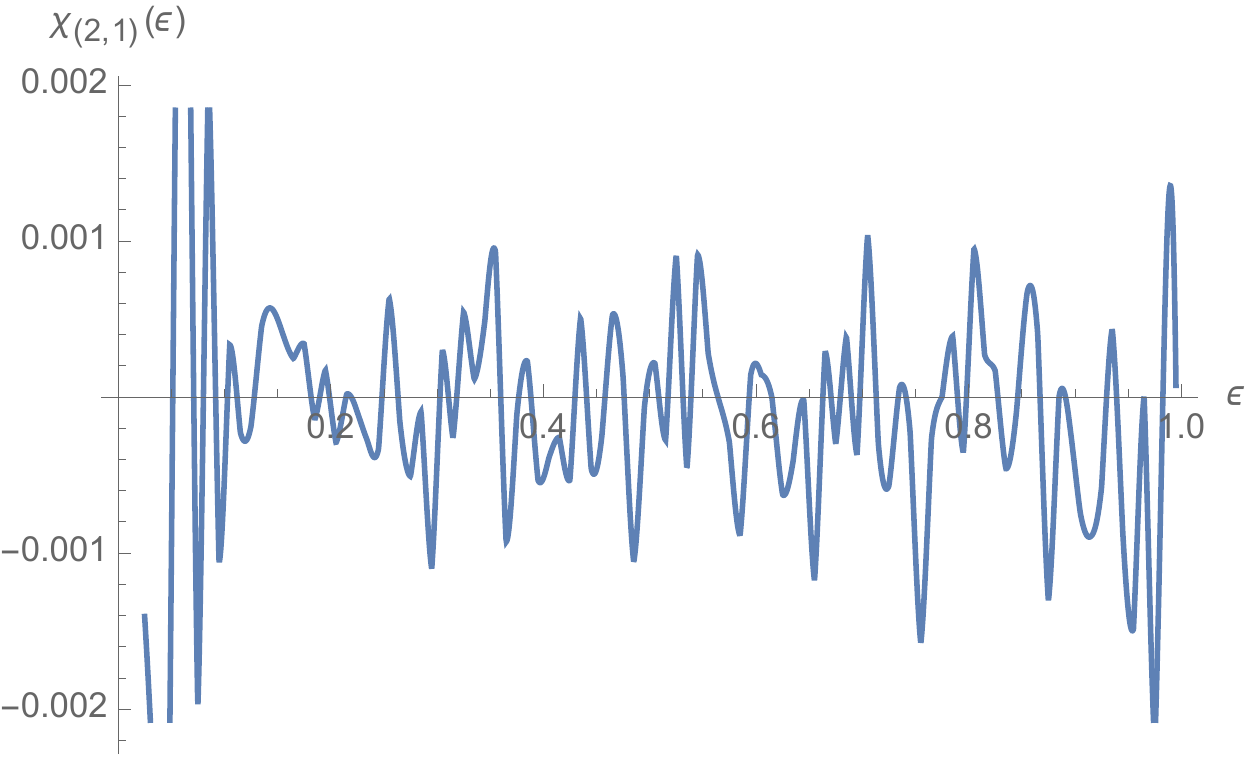}
    \caption{Residuals from a fit of the Lovas-Andai-type formula $\tilde{\chi}_{2,1}(\varepsilon)$ to its estimation based on
sixty million randomly generated (with respect to $k=1$ induced measure) $4 \times 4$ density matrices}
    \label{fig:randominducefitk1}
\end{figure}
(We can now further, using $\tilde{\chi}_{2,1}(\varepsilon)$, find the $k=1$ separability probability based on the previously-discussed operator monotone function $x \rightarrow \sqrt{x}$. The result we obtain is $4427-\frac{131072}{3 \pi ^2} =19 \cdot 233 -\frac{2^{17}}{3 \pi^2} \approx 0.209939$.)

Proceeding onward to the $k=2$ case, still in the complex domain ($\mathbb{C}$), we have
\begin{equation} \label{interpolationk2}
\frac{\int_{-1}^1  \int_{-1}^x\tilde{\chi}_{2,2} (\sqrt{\frac{1-x}{1+x}}  \sqrt{\frac{1+y}{1-y}})(1-x^2)^4 (1-y^2)^4 (x-y)^2 \mbox{d} y \mbox{d} x}{\int_{-1}^1  \int_{-1}^x(1-x^2)^4 (1-y^2)^4 (x-y)^2 \mbox{d} y \mbox{d} x}=  \frac{259}{442},
\end{equation}
agreeing with (\ref{qubitinduced}), where, now, 
\begin{equation*}
 \tilde{\chi}_{2,2}(\varepsilon)    =\frac{1}{5} \varepsilon ^2 \left(-\varepsilon ^6+8 \varepsilon ^4-18 \varepsilon ^2+16\right).
\end{equation*}
(We can now, using $\tilde{\chi}_{2,2}(\varepsilon)$ find the $k=2$ separability probability based, once again, on the previously-discussed operator monotone measure. The result we obtain is $-\frac{1713917}{3}+\frac{26642219008}{4725 \pi ^2} = 
-\frac{61 \cdot 28097}{3}+\frac{2^{26} \cdot 397}{3^2 \cdot 5^2 \cdot 19 \pi^2} \approx 0.399947$.)

Moving from the complex to quaternionic domain ($\mathbb{H}$), again for $k=1$, we have
\begin{equation} \label{interpolationk3}
\frac{\int_{-1}^1  \int_{-1}^x\tilde{\chi}_{4,1} (\sqrt{\frac{1-x}{1+x}}  \sqrt{\frac{1+y}{1-y}})(1-x^2)^5 (1-y^2)^5 (x-y)^4 \mbox{d} y \mbox{d} x}{\int_{-1}^1  \int_{-1}^x(1-x^2)^5 (1-y^2)^5 (x-y)^4 \mbox{d} y \mbox{d} x}=  \frac{3736}{22287},
\end{equation}
agreeing with (\ref{quaterbitinduced}), where, now, we employ 
\begin{equation}
 \tilde{\chi}_{4,1}(\varepsilon)    = \frac{1}{21} \varepsilon ^4 \left(-9 \varepsilon ^6+55 \varepsilon ^4-125 \varepsilon ^2+100\right).
\end{equation}
(The separability probability, based on $\tilde{\chi}_{4,1}(\varepsilon)$, in the noted operator monotone case is 
$\frac{27637}{168}-\frac{50 \pi ^2}{3} = \frac{29 \cdot 953}{2^3 \cdot 3 \cdot 7} 
-\frac{2 \cdot 5^2 \pi^2}{3} \approx 0.0125457$.) 

We remark that the two-{\it rebit} ($d=1$)  functions $\tilde{\chi}_{1,k}(\varepsilon)$, and more generally $\tilde{\chi}_{d,k}(\varepsilon)$, for {\it odd} $d$, appear to be of considerably more complicated non-polynomial form, involving inverse hyperbolic, logarithmic and polylogarithmic functions.)
\subsection{General $\tilde{\chi}_{d,k}(\varepsilon)$ formula}
It now seems clear that to obtain--in our extended Lovas-Andai framework--an induced measure-based separability/PPT probability ($\mathcal{P}_{sep/PPT}(d,k)$) in the real ($\mathbb{R}$), complex ($\mathbb{C}$) or 
quaternionic  ($\mathbb{H}$) domain, we must set the exponent ($d$) of the
$(x-y)$ terms in the numerators and denominators   to 1, 2 or 4, respectively. While to obtain a specific 
$k$-induced measure result, we must take the exponents of the $(1-x^2)$ and $(1-y^2)$ terms to be $d+k$. In other words, we have the general ($(d,k)$-parameterized) 
formula $\mathcal{P}_{sep/PPT}(d,k)$ given by (\ref{interpolationGeneral}). 
If we replace the $d+k$ terms  in (\ref{interpolationGeneral}) by $-\frac{d}{4}+k$, we obtain the separability/PPT-probability function in the operator monotone function $\sqrt{x}$ case.
\subsection{Analytical background}
Now, let us indicate the general manner in which we obtained the three specific indicated new functions $\tilde{\chi}_{2,1}(\varepsilon)$, $\tilde{\chi}_{2,2}(\varepsilon)$ 
and  $\tilde{\chi}_{4,1}(\varepsilon)$ above. 
In this direction, we have for the complex case, $d=2$, the  general induced measure formula
\begin{equation} \label{(d,k)-parameterized}
 \tilde{\chi}_{2,k}(\varepsilon)   =\frac{\left(-k+\varepsilon ^2-3\right) \left(1-\varepsilon ^2\right)^{k+1}+k+3}{k+3},
\end{equation}
with the denominator of (\ref{interpolationGeneral}) evaluating to $\frac{\pi  \Gamma (k+3)^2}{2 \Gamma \left(k+\frac{7}{2}\right) \Gamma
   \left(k+\frac{9}{2}\right)}$. Alternatively,
\begin{equation}
\tilde{\chi}_{2,k}(\varepsilon)  =    \frac{(k+2) \varepsilon ^2 \left(\sum _{j=0}^{k-1} \left(1-\varepsilon ^2\right)^{k-j}
     {_2F_1}\left(1,-j-1;k+4;\varepsilon ^2\right)+2 \,
   _2F_1\left(1,-k-1;k+4;\varepsilon ^2\right)\right)}{k+3}.
\end{equation}
(For a formal proof by C. Dunkl of this relation, see App. \ref{doublesums}.)
We, then, have for $k=-\frac{5}{2}$,
\begin{equation}
 \tilde{\chi}_{2,-\frac{5}{2}}(\varepsilon)   = 2 \left(\frac{\varepsilon ^2-\frac{1}{2}}{\left(1-\varepsilon
   ^2\right)^{3/2}}+\frac{1}{2}\right),
\end{equation}
which we can interestingly use to replace $\tilde{\chi}_{2}(\varepsilon) $ in (\ref{interpolation}), giving us (again setting $\eta=-\frac{1}{2}$) a result 
now of $\frac{21 \pi -64}{21 \pi } \approx 0.0299127$ to compare (in the induced measure framework) with the previously-given ($x \rightarrow \sqrt{x}$) operator monotone result of 
$1 -\frac{256}{27 \pi^2}  \approx 0.0393251$. (We can closely approximate this result using $\tilde{\chi}_{2,-2.86051355}(\varepsilon)$.)
Additionally, Charles Dunkl has found--employing rational interpolation--the (quaternionic, $d=4$) counterpart (and also for $d=6$),
\begin{equation}
\tilde{\chi}_{4,k}(\varepsilon)=    \left(\frac{2 \left(2 k^2+14 k+21\right) \varepsilon ^4}{(k+5) (k+6)}-\frac{6 (k+3) \varepsilon
   ^6}{(k+6) (k+7)}-(k+1) \varepsilon ^2-1\right) \left(1-\varepsilon ^2\right)^{k+1}+1.
\end{equation}
For $k=0$, we recover the previously-reported Hilbert-Schmidt formula of $\tilde{\chi}_{4}(\varepsilon)=\frac{1}{35} \varepsilon ^4 \left(15 \varepsilon ^4-64 \varepsilon ^2+84\right)$ 
\cite[sec. VI]{slater2017master}. The corresponding formula for $k=1$ is (\ref{interpolationk3}).
\subsection{Methodology employed}
To further elaborate upon the general methodology employed to obtain the above results, we refer to the analyses and notation employed in \cite[sec. VII]{slater2017master}. We must, again, perform the constrained integrations  presented there, but now, additionally,  for induced measure of order $k \neq 0$, we must
multiply both the (numerator and denominator) integrands by the $k$-th power of 
$\left(\left(r_{14}^2-1\right) \left(r_{23}^2-1\right)-r_{24}^2\right) \varepsilon ^2$. This term is the relevant factor in the determinant
of the $4 \times 4$ density matrices (having three pairs of nullified entries) employed in the cited reference. The additional determinantal factors are all positive and not functions of the $r$'s, and would cancel, so can be ignored in the computations.

To be more specific, in these regards, in \cite[sec. VII]{slater2017master} we employed the set of constraints (imposing--in quantum-information-theoretic terms--the positivity of the density matrix and its partial transpose),
\begin{equation}
r_{23}^2<1\land \left(r_{14}^2-1\right) \left(r_{23}^2-1\right)>r_{24}^2\land r_{23}^2
   \left(\varepsilon ^2 r_{14}^2-1\right)>\varepsilon ^2 \left(\varepsilon ^2
   r_{14}^2+r_{24}^2-1\right).   
\end{equation}
Then, subject to these constraints, we had to integrate the jacobian (corresponding to the hyperspherical parameterization
of the three off-diagonal non-nullified entries of the density matrix) $\left(r_{14} r_{23} r_{24}\right){}^{d-1}$
over the unit cube $[0,1]^3$. Dividing the result of the integration by
\begin{equation} \label{denominator}
\frac{\pi  4^{-d} \Gamma \left(\frac{d}{2}+1\right)^2}{d^3 \Gamma
   \left(\frac{d+1}{2}\right)^2},
\end{equation}
yielded the desired $\tilde{\chi_d} (\varepsilon )$. (If we were to take $r_{24}=0$, and a jacobian of 
$\left(r_{14} r_{23}\right){}^{d-1}$, we would revert to the X-states setting, and obtain simply 
$\varepsilon^d$ as the corresponding  function.)

This last result (\ref{denominator}) was obtained by integrating the same jacobian 
$\left(r_{14} r_{23} r_{24}\right){}^{d-1}$ over the unit cube,
subject to the constraints (imposing the positivity of the density matrix),
\begin{equation}
r_{23}^2<1\land \left(r_{14}^2-1\right) \left(r_{23}^2-1\right)>r_{24}^2.
\end{equation}
So to reiterate, to move on to the more general induced measure setting (that is, $k \ne 0$), we must multiply both the indicated (numerator and denominator) integrands by the $k$-th power of 
$\left(\left(r_{14}^2-1\right) \left(r_{23}^2-1\right)-r_{24}^2\right) \varepsilon ^2$. The Hilbert-Schmidt ($k=0$) denominator
integration result (\ref{denominator}), then, generalizes to 
\begin{equation}
 \frac{\Gamma \left(\frac{d}{2}\right)^3 \varepsilon ^{2 k} \Gamma (k+1) \Gamma
   \left(\frac{d}{2}+k+1\right)}{8 \Gamma (d+k+1)^2}.   
\end{equation}
\subsection{Extended master formula investigation}
Our goal now is the development of a still more general  Lovas-Andai ``master formula'' for $\tilde{\chi}_{d,k}(\varepsilon)$ than has been so far reported for  $\tilde{\chi}_{d}(\varepsilon) \equiv \tilde{\chi}_{d,0}(\varepsilon)$ in \cite[sec. VII.A]{slater2017master}, that is, (\ref{MasterFormula}) above.
An expression for the anticipated (induced measure) extended master formula is as the sum of (\ref{onehalf}),
(reducing to one-half of (\ref{MasterFormula}) for $k=0$)
and the two-dimensional integral of the product of 
\begin{equation} \label{integrand}
  \frac{1}{\Gamma \left(\frac{d}{2}\right)^2 \Gamma (k+1) \Gamma
   \left(\frac{d}{2}+k+1\right)}  
\end{equation}
and
\begin{equation}
Y^{d-1} \left(\frac{1}{r_{14} \varepsilon }\right){}^{d+1} \left(1-r_{14}^2 \varepsilon
   ^2\right){}^{d/2} \left(\left(1-\frac{1}{r_{14}^2}\right) Y^2-r_{14}^2+1\right){}^k
   \left(r_{14}^2 \varepsilon ^2-Y^2\right){}^{d/2}
\end{equation}
and
\begin{equation}
\, _2\tilde{F}_1\left(\frac{d}{2},-k;\frac{d+2}{2};\frac{\left(r_{14}^2 \varepsilon
   ^2-1\right) \left(Y^2-r_{14}^2 \varepsilon ^2\right)}{\left(r_{14}^2-1\right) \varepsilon ^2
   \left(Y^2-r_{14}^2\right)}\right)   .  
\end{equation}
The two-dimensional domain of integration is 
\begin{equation}
r_{14} \in [0,1], \hspace{.25in}  Y \in [\varepsilon   r_{14}, \varepsilon^2   r_{14}]  .
\end{equation}
The result of this integration must also, as (\ref{onehalf}) does, equal $\frac{1}{2}$  of the master formula (\ref{MasterFormula}) result for $k=0$.
(Questions pertaining to these last discussed issues have been posted at https://mathematica.stackexchange.com/questions/171351/evaluate-over-a-two-dimensional-domain-the-integral-of-hypergeometric-based-f
and https://math.stackexchange.com/questions/2744828/find-five-parameter-values-for-a-3-tildef-2-function-yielding-five-pCFDinolynomi  .)   C. Dunkl has 
been able to solve this two-dimensional integration problem (with a change-of-variables, I had suggested, in the integrand), with the
result (App.~\ref{CFDintegration}) given near the end of sec.~\ref{Introduction}.

In \cite[sec. VII.A]{slater2017master}, we reported a formula there (74), which when
multiplied by the ``master  Lovas-Andai'' [Hilbert-Schmidt, $k=0$] formula for $\tilde{\chi}_{d}(\varepsilon)$, (70) there, and appropriately
normalized (by (75) there) would yield the $d$-specific separability/PPT-probability. We have now obtained the 
induced measure extension of that Hilbert-Schmidt formula (74)  (thus, reducing to it for $k=0$). It takes the form
\begin{equation} \label{unnormalized}
(-1)^d 2^{5 d+4 k+3} \left(t^2-1\right)^d t^{-4 d-2 k-3} \Gamma (3 d+2 k+2)^2 \,
   _2\tilde{F}_1\left(3 d+2 k+2,3 d+2 k+2;6 d+4 k+4;1-\frac{1}{t^2}\right).    
\end{equation}
For its normalization factor, corresponding to \cite[eq. (75)]{slater2017master}, we have found 
\begin{equation} \label{normalization}
2^{5 d+4 k+2} \Gamma (d+1) \Gamma (d+k+1) \Gamma (3 d+2 k+2)^2 \, _3\tilde{F}_2(3 d+2
   k+2,3 d+2 k+2,d+1;6 d+4 k+4,2 d+k+2;1).    
\end{equation}
As a check on this normalization formula, in the Hilbert-Schmidt ($k=0$) case, we see that it agrees with \cite[eq. (75)]{slater2017master}
for $d=6$, both formulas yielding $\frac{33554432}{4854694845}$. (For a related analysis of Dunkl, see App.~\ref{laststeps}.) 
For $d \bmod 4=0$, the normalization factor reduces to
\begin{equation}
\frac{\Gamma \left(\frac{d+1}{2}\right) 2^{\left\lfloor \frac{d}{4}\right\rfloor
   +\frac{11 d}{4}+2 k} \Gamma (d+k+1)^2 \Gamma \left(\frac{3 d}{2}+k+1\right) \Gamma
   \left(\frac{5 d}{4}+k+\left\lfloor \frac{d}{4}\right\rfloor
   +\frac{3}{2}\right)}{\Gamma \left(\frac{3 (d+1)}{2}+k\right)^2 \Gamma \left(\frac{5
   d}{2}+2 k+2\right)} .   
\end{equation}
In the operator monotone
($\sqrt{x}$) case at hand, the analogous (hypergeometric-free) normalization factor to (\ref{normalization}) is 
\begin{equation}
    \frac{4^{\frac{d}{4}+k} \Gamma \left(\frac{d}{2}+\frac{1}{2}\right) \Gamma
   \left(-\frac{d}{4}+k+1\right)^2 \Gamma \left(\frac{d}{4}+k+1\right)}{\Gamma (2 k+2)
   \Gamma \left(\frac{d}{4}+k+\frac{3}{2}\right)}.
\end{equation}
We have been able to use the two results ((\ref{unnormalized}), (\ref{normalization})) to see that in agreement with a previously-reported formula \cite[eq. (62)]{slater2016formulas} (given above in sec.~\ref{Determinantal}) 
 that (\ref{onehalf}) accounts for that part of the separability probability for which $|\rho^{PT}| \geq |\rho|$, and the formula of Dunkl (\ref{otherhalf}) for the complementary separability probability. (We have specifically verified this assertion in the $d=2,k=1$ case for which we have $\tilde{\chi}_{2,1}(\varepsilon)=\frac{1}{4} \varepsilon ^2 \left(3-\varepsilon ^2\right)^2$, given by (\ref{d=2k=1}). This decomposes into the sum of the Lovas-Andai separability functions $\frac{1}{20} \varepsilon ^2 \left(\varepsilon ^4-6 \varepsilon ^2+15\right)$ and $\frac{3 \varepsilon ^2}{2}+\frac{1}{5} \left(\varepsilon ^2-6\right) \varepsilon ^4$. These yield, respectively, separability probabilities of $\frac{45}{286}$, for the first complementary part, and  $\frac{7}{26}$, for the second, summing, as known, to $\frac{61}{143}$.)
\section{Concluding remarks} \label{Concluding}
An interesting question is whether it is possible to obtain the Bures (minimal monotone) two-qubit separability probability within the 
Lovas-Andai framework, in a somewhat similar manner to how we have extended this framework to the induced measures. Now, in the Bures
case, rather than incorporating simply powers of the determinants of the density matrices into the required integrations, we must include the 
reciprocal ($R$) of the product ($P$) of the square root of the determinant and the product term $\Pi_{j<k}^4 (\lambda_i+\lambda_j)$, where the $\lambda$'s  
are the eigenvalues of the density matrices \cite[eq. (3.18)]{sommers2003bures}. We were able to employ the symmetric reduction of this last term and re-express it in terms of the traces of matrix powers of the density matrices, and, thus, of the entries of the density matrix. Unfortunately, we have not been able to
perform the requisite integrations. (We should note that both the Hilbert-Schmidt and operator monotone ($\sqrt{x}$) measures lead to separability probabilities invariant over the Bloch radii of the qubit subsystems, while the Bures measure has been demonstrated in \cite[Fig. 31]{slater2015bloch} to not have such behavior. Lovas and Andai have commented \cite[sec. 6]{lovas2017invariance} that in this operator monotone case, the associated  probability is determined by the structure of the surface of the unit ball.)

Of course, it would be most desirable to  rigorously derive the Hilbert-Schmidt/Lebesgue separability/PPT-probabilities for the 35- and 63-dimensional convex sets of qubit-qutrit and qubit-qudit states, among others,  examined above. But, given that the  Hilbert-Schmidt 
separability probability of $\frac{8}{33}$ for the 15-dimensional convex set of two-qubit  states has itself proved highly formidable to well establish \cite{slater2017master,lovas2017invariance,milz2014volumes,fei2016numerical,shang2015monte,slater2013concise,slater2012moment,slater2007dyson} (though not yet formally proven), 
it seems that major advances would be required to achieve such a goal 
in these still higher-dimensional settings (and, thus, confirm or reject the conjectures above).

Implicit in the analytical approach pursued here has been the clearly yet unverified assumption that the separability/PPT-probabilities will continue to be {\it rational-valued} for the higher-dimensional systems, as they have, remarkably, been found to be in the $4 \times 4$ setting.

We can, of course, as future research, continue our various simulations 
reported above of random density matrices, hoping to obtain further accuracy in  separability/PPT-probability estimates. One relevant issue of interest would then be the trade-off between the use of  increased precision in the random normal variates employed (we have so far used the Mathematica default precision option), and the presumed consequence, then, of decreased number of  variates to be generated.
\subsection{Casimir invariants} \label{Casimir}
One of our goals here  has been to determine if we could use the $N=4$ results \cite{slater2017master,lovas2017invariance,milz2014volumes,fei2016numerical,shang2015monte,slater2013concise,slater2012moment,slater2007dyson} to gain insight into the $N>4$ counterparts, and, more specifically, 
if certain analytical properties continue to hold. We found some encouragement for undertaking such a course from the research reported in 
 \cite{slater2016invariance}. There, evidence was provided that a most interesting common characteristic is 
shared by two-qubit ($N=4$), qubit-qutrit $(N=6$), qubit-qudit ($N=8$, specifically) and two-qutrit ($N=9$) systems. That is, the associated
(HS) separability/PPT probabilities hold constant over the {\it Casimir invariants} \cite{gerdt20116,byrd2003characterization} of both their subsystems
(such as the lengths of the Bloch radii of the reduced qubit subsystems) (cf. \cite[Corollary 2]{lovas2017invariance}). (A Casimir invariant is a  distinguished element of the center of the universal enveloping algebra of a Lie algebra \cite{gerdt20116}.) 

It would be of interest to computationally employ such apparent invariance (formally proved by Lovas and Andai
\cite[Corollary 2]{lovas2017invariance} in the two-rebit $\frac{29}{64}$ case) in strategies to ascertain these various separability/PPT-probabilities. However, we have yet to find
an effective manner of doing so (even after setting the Casimir invariants to zero, leading to 
lower-dimensional settings). (In our paper, ``Two-qubit separability probabilities as joint functions of the Bloch radii of the qubit subsystems'' \cite{slater2016two}, we observed a relative repulsion effect between the Casimir invariants of the two reduced systems of several forms of bipartite states.)

Let us, in these regards, also indicate the interesting paper of Altafini, entitled ``Tensor of coherences parametrization of multiqubit density operators for entanglement characterization'' \cite{altafini2004tensor}. In it, he applies the term ``partial quadratic Casimir invariant'' in relation to reduced density matrices. He notes that a quadratic Casimir invariant can be regarded as the specific form ($q = 2$) of Tsallis entropy. Further, he remarks that ``partial transposition is a linear norm preserving operation: $\mbox{tr}(\rho^2) = \mbox{tr}((\rho^{T_1})^2) =  \mbox{tr}((\rho^{T_2})^2)$. Hence entanglement violating PPT does {\it not modify} the quadratic Casimir invariants of the density and the necessary [separability] conditions $[\mbox{tr}(\rho^2_A) \geq \mbox{tr}(\rho^2), \mbox{tr}(\rho^2_B) \geq \mbox{tr}(\rho^2)] $, are 
{\it insensible} to it'' (emphasis added).

Let us, relatedly, indicate the pair of formulas (cf. (\ref{ZSComplex}), (\ref{AndaiComplex}))
\begin{equation} \label{MZ1}
  V_{HS}^{2 \times m}(r)=   V_{HS}^{2 \times m}(0) (1-r^2)^{2 (m^2-1)}
\end{equation}
and
\begin{equation} \label{MZ2}
 V_{HS}^{2 \times m}(0)   = \sqrt{m} \cdot 2^{6 m^2 -m -\frac{23}{2}} \cdot \pi^{2 m^2-m -\frac{3}{2}} \cdot 
 \frac{\Pi_{k=1}^{2 m} \Gamma{(k)} \cdot \Gamma{(\frac{1}{2}+2 m^2)}}{\Gamma{(4 m^2)} \cdot \Gamma{(-1+2 m^2)}}
\end{equation}
that Milz and Strunz conjectured for the Hilbert-Schmidt volume of the $2\times m$ qubit-qudit states \cite[eqs. (27), (28)]{milz2014volumes}, as a function of the Bloch radius ($r$) of the qubit subsystem. (These appear to have been confirmed for the two-qubit 
[$m=2$] case by the analyses of Lovas and Andai 
\cite[Cor. 1]{lovas2017invariance}.)
\appendix
\section{Qubit-qutrit induced measures analyses} \label{Qubitqutritinduced}
Let us now investigate qubit-{\it qutrit} scenarios in which the measure employed is {\it not} that induced by tracing over a $K$-dimensional environment, where $K=6$, $k=K-6=0$, as in the Hilbert-Schmidt case, but with $K \neq 6$, $k \neq 0$. 

For the corresponding induced measure two-qubit cases, we reported, among others, the formula \cite[eq. (2)]{slater2015formulas} \cite[eq. (4)]{slater2016formulas},
\begin{equation} \label{qubitinduced}
 P^{2-qubits}_k=1-\frac{3 \cdot 4^{k+3} (2 k (k+7)+25) \Gamma{(k+\frac{7}{2}}) \Gamma{(2 k+9)}}{\sqrt{\pi} \Gamma{(3 k+13)}}.
\end{equation}
\subsection{$k=2$, $K=8$} \label{k2K8}
In the two-qubit case for $k=2$, the formula (\ref{qubitinduced})  gives $\frac{259}{442}=\frac{7 \cdot 37}{2 \cdot 13 \cdot 17} \approx 
0.585973$ (see also (\ref{interpolationk2}) below). Now, of 150,000,000 randomly-generated  qubit-qutrit density matrices with the indicated $k=2$ measure, 23,721,307 had PPT's, yielding an estimated separability probability of 0.15814205.

Among these 23,721,307,  only 171 of them passed the further test for {\it separability from spectrum} presented by Johnston \cite[Thm. 1]{johnston2013separability}. That is, only for 
these 171, did the condition hold that $\lambda_1< \lambda_5 +2 \sqrt{\lambda_4 \lambda_{6}}$, where the $\lambda$'s are the six ordered eigenvalues of 
the density matrices, with $\lambda_1$ being the greatest (cf. \cite[App. A]{slater2017master}).

\subsection{$k=1$, $K=7$} \label{k1K7}
In the two-qubit case for $k=1$, the formula (\ref{qubitinduced})  gives $\frac{61}{143}=\frac{61}{11 \cdot 13} \approx 0.4265734$ 
 (see also (\ref{interpolation2}) below). 
Of 171,000,000 randomly-generated qubit-qutrit density matrices for $k=1$, 13,293,906 had PPT's, yielding an estimated separability probability of 0.0777402.
Among these 13,293,906,  only 19 passed the previously-noted (Johnston) test for  separability from spectrum.

\subsection{$k=-1$, $K=5$}
In the two-qubit case with $k=-1$, the
formula (\ref{qubitinduced}) yields $\frac{1}{14}=\frac{1}{2 \cdot 7}$ \cite[sec. III]{slater2015formulas}.
Now, of 294,000,000 randomly-generated such $6 \times 6$ density matrices, 1,435,605 had PPT's, giving 0.00488301, as a separability probability.
\subsection{$k=-2$, $K=4$}
In the two-qubit case with $k=-2$, the associated separability probability must be null, since the ranks of the density matrices are not greater than the dimensions of the reduced systems  \cite{ruskai2009bipartite}. (The value zero is, in fact, yielded by the two-qubit formula (\ref{qubitinduced}) for $k=-2$.) Now, of 330,000,000 randomly-generated  $6 \times 6$ density matrices with $k=-2$, 55,037 had PPT's, giving 0.000166779, as an estimated separability probability. 

At the present stage of our research, we are reluctant to advance specific conjectures for the four immediately preceding
induced-measure qubit-qutrit analyses ($k=2,1,-1,-2$).

\section{Appendices of Charles Dunkl pertaining to Lovas-Andai formula extension to induced measures} \label{CFD}
\subsection{Evaluation of double sums} \label{doublesums}
The goal is a closed form for%
\begin{align*}
\chi_{d,k}\left(  \varepsilon^{2}\right)    & =\frac{\Gamma\left(
1+d+k\right)  ^{2}\Gamma\left(  1+d\right)  }{2\Gamma\left(  1+\frac{d}%
{2}+k\right)  \Gamma\left(  1+\frac{d}{2}\right)  ^{2}\Gamma\left(
1+\frac{3d}{2}+k\right)  }\varepsilon^{d}\\
& \times\{2~_{3}F_{2}\left(
\genfrac{}{}{0pt}{}{-d/2-k,d/2,d}{1+d/2,1+k+3d/2}%
;\varepsilon^{2}\right)  \\
& +\sum_{j=0}^{k-1}\frac{\left(  d\right)  _{k-j}}{\left(  \frac{d}%
{2}+1\right)  _{k-j}}\left(  1-\varepsilon^{2}\right)  ^{k-j}~_{3}F_{2}\left(
%
\genfrac{}{}{0pt}{}{-d/2-j,d/2,d+k-j}{1+d/2+k-j,1+k+3d/2}%
;\varepsilon^{2}\right)  \}
\end{align*}
when $d$ is even. Suppose $d=2m$; the conjecture is that $\chi_{2m,k}\left(
\varepsilon^{2}\right)  =\left(  1-\varepsilon^{2}\right)  ^{k+1}p_{m}\left(
k,\varepsilon^{2}\right)  +1$ where $p_{m}$ is a polynomial of degree $2m+1$
in $\varepsilon^{2}$ and the coefficients are rational functions of $k$.

Start with $d=2$. The sum specializes to%
\begin{gather*}
\chi_{2,k}\left(  \varepsilon^{2}\right)  =\frac{k+2}{k+3}\varepsilon^{2}\\
\times\{2~_{2}F_{1}\left(  -1-k,1;k+4;\varepsilon^{2}\right)  +\sum
_{j=0}^{k-1}\left(  1-\varepsilon^{2}\right)  ^{k-j}~_{2}F_{1}\left(
-1-j,1;k+4;\varepsilon^{2}\right)  \}.
\end{gather*}
Our approach is to replace $1-\varepsilon^{2}$ by $t$ and express the
expression as a polynomial in $t$ (the claim is that the coefficient of
$t^{s}$ is zero for $1\leq s\leq k$).

\begin{lemma}
Suppose $n=1,2,3,\ldots$ then%
\[
\left(  1-t\right)  ~_{2}F_{1}\left(  -n,1;c;1-t\right)  =\frac{c-1}%
{c+n-1}-\frac{\left(  c-2\right)  \left(  c-1\right)  t}{\left(  c+n-2\right)
\left(  c+n-1\right)  }~_{2}F_{1}\left(  -n,1;3-n-c;t\right)  .
\]

\end{lemma}

\begin{proof}
Expand the left side%
\begin{align*}
S  & =\sum_{j=0}^{n}\frac{\left(  -n\right)  _{j}}{\left(  c\right)  _{j}}%
\sum_{i=0}^{j+1}\frac{\left(  j+1\right)  !}{i!\left(  j+1-i\right)  !}\left(
-t\right)  ^{i}\\
& =\sum_{j=0}^{n}\frac{\left(  -n\right)  _{j}}{\left(  c\right)  _{j}}%
+\sum_{i=1}^{n+1}\frac{\left(  -t\right)  ^{i}}{i!}\sum_{j=i-1}^{n}%
\frac{\left(  -n\right)  _{j}\left(  j+1\right)  !}{\left(  c\right)
_{j}\left(  j+1-i\right)  !};
\end{align*}
the first term sums (by Chu-Vandermonde), in the second part let $j=s+i-1$ so
that $0\leq s\leq n-i+1$; then%
\begin{align*}
S  & =\frac{\left(  c-1\right)  _{n}}{\left(  c\right)  _{n}}-t\sum
_{i=1}^{n+1}\frac{\left(  -n\right)  _{i-1}i!}{\left(  c\right)  _{i-1}%
i!}\left(  -t\right)  ^{i-1}\sum_{s=0}^{n-i+1}\frac{\left(  -n+i-1\right)
_{s}\left(  i+1\right)  _{s}}{\left(  c+i-1\right)  _{s}s!}\\
& =\frac{c-1}{c+n-1}-t\sum_{i=1}^{n+1}\frac{\left(  -n\right)  _{i-1}%
i!}{\left(  c\right)  _{i-1}i!}\left(  -t\right)  ^{i-1}\frac{\left(
c+i-1-i-1\right)  _{n-i+1}}{\left(  c+i-1\right)  _{n-i+1}}\\
& =\frac{c-1}{c+n-1}-t\frac{\left(  c-2\right)  _{n}}{\left(  c\right)  _{n}%
}\sum_{u=0}^{n}\frac{\left(  -n\right)  _{u}}{\left(  3-n-c\right)  _{u}}%
t^{u};
\end{align*}
because $\left(  c-2\right)  _{n-i+1}=\left(  -1\right)  ^{s}\left(
c-2\right)  _{n}/\left(  3-n-c\right)  _{u}$ where $u=i-1$ (the sum is over
$0\leq u\leq n$).
\end{proof}

Now apply the Lemma with $n=j+1$ and $c=k+4$ obtaining%
\[
\varepsilon^{2}~_{2}F_{1}\left(  -1-j,1;k+4;\varepsilon^{2}\right)
=\frac{k+3}{k+4+j}-t\frac{\left(  k+2\right)  \left(  k+3\right)  }{\left(
k+3+j\right)  \left(  k+4+j\right)  }~_{2}F_{1}\left(  -1-j,1;-2-j-k;t\right)
.
\]
Denote this sum by $S_{j}\left(  t\right)  $. The aim is to determine the
coefficients of $t^{s}$ in the sum (omit the leading factor $\dfrac{k+2}{k+3}$
for now)%
\[
2S_{k}\left(  t\right)  +\sum_{j=0}^{k-1}t^{k-j}S_{j}\left(  t\right)  .
\]
The constant term appears only in $2S_{k}\left(  t\right)  $, and equals
$\dfrac{k+3}{2k+4}=\dfrac{k+3}{2\left(  k+2\right)  }$, $\chi_{2,k}\left(
1\right)  =1$. Fix $u$ with $1\leq k-u\leq k$ ($0\leq u\leq k-1$) and find the
coefficient of $t^{k-u}$ in the sum. Denote $coef\left(  f\left(  t\right)
,t^{i}\right)  $ to be the coefficient of $t^{i}$ in $f\left(  t\right)  $.
The above shows (for $1\leq i\leq j+2$)
\[
coef\left(  S_{j}\left(  t\right)  ,t^{i}\right)  =-\frac{\left(  k+2\right)
\left(  k+3\right)  }{\left(  k+3+j\right)  \left(  k+4+j\right)  }%
\frac{\left(  -1-j\right)  _{i-1}}{\left(  -2-j-k\right)  _{i-1}}
\]
Compute%
\begin{align*}
2coef\left(  S_{k}\left(  t\right)  ,t^{k-u}\right)  +\sum_{j=u}%
^{k-1}coef\left(  S_{j}\left(  t\right)  ,t^{j-u}\right)    & =-\frac
{k+3}{2k+3}\frac{\left(  -1-k\right)  _{k-u-1}}{\left(  -2-2k\right)
_{k-u-1}}+\frac{k+3}{k+4+u}\\
& -\sum_{j=u+1}^{k-1}\frac{\left(  k+2\right)  \left(  k+3\right)  \left(
-1-j\right)  _{j-u-1}}{\left(  k+3+j\right)  \left(  k+4+j\right)  \left(
-2-j-k\right)  _{j-u-1}}.
\end{align*}
Change summation index $j=u+1+s.$ Claim for $m=0,1,2,\ldots$%
\begin{equation}
\frac{1}{k+4+u}-\sum_{s=0}^{m}\frac{\left(  k+2\right)  \left(  -u-s-2\right)
_{s}}{\left(  k+4+u+s\right)  \left(  k+5+u+s\right)  \left(  -3-u-s-k\right)
_{s}}=\frac{\left(  u+3\right)  _{m+1}}{\left(  k+4+u\right)  _{m+2}%
}.\label{sum1}%
\end{equation}
Proceeding by induction suppose the formula is true (it holds for $m=-1)$ for
some $m$, then%
\begin{align*}
& \frac{\left(  u+3\right)  _{m+1}}{\left(  k+4+u\right)  _{m+2}}%
-\frac{\left(  k+2\right)  \left(  -u-m-3\right)  _{m+1}}{\left(
k+5+u+m\right)  \left(  k+6+u+m\right)  \left(  -4-u-m-k\right)  _{m+1}}\\
& =\frac{\left(  u+3\right)  _{m+1}}{\left(  k+4+u\right)  _{m+2}}%
-\frac{\left(  k+2\right)  \left(  u+3\right)  _{m+1}}{\left(  k+4+u\right)
_{m+3}}=\frac{\left(  u+3\right)  _{m+1}\left\{  \left(  k+6+u+m\right)
-(k+2)\right\}  }{\left(  k+4+u\right)  _{m+3}}\\
& =\frac{\left(  u+3\right)  _{m+1}\left(  u+4+m\right)  }{\left(
k+4+u\right)  _{m+3}}=\frac{\left(  u+3\right)  _{m+2}}{\left(  k+4+u\right)
_{m+3}};
\end{align*}
this completes the induction (note the reversal $\left(  -a-m\right)
_{m+1}=\left(  -1\right)  ^{m+1}\left(  a\right)  _{m+1}$). Apply the formula
with $m=k-u-2$ (since $j=u+s+1\leq k-1$)%
\begin{align*}
-\frac{k+3}{2k+3}\frac{\left(  -1-k\right)  _{k-u-1}}{\left(  -2-2k\right)
_{k-u-1}}+\frac{\left(  k+3\right)  \left(  u+3\right)  _{k-u-1}}{\left(
k+4+u\right)  _{k-u}}  & =-\frac{\left(  k+3\right)  \left(  u+3\right)
_{k-u-1}}{\left(  2k+3\right)  \left(  k+4+u\right)  _{k-u-1}}+\frac{\left(
k+3\right)  \left(  u+3\right)  _{k-u-1}}{\left(  k+4+u\right)  _{k-u}}\\
& =0,
\end{align*}
by use of the reversing formula (for $\left(  -a\right)  _{k-u-1}$) and
$\left(  2k+3\right)  \left(  k+4+u\right)  _{k-u-1}=\left(  k+4+u\right)
_{k-u}$. Thus $\chi_{d,k}\left(  1-t\right)  =1+c_{1}t^{k+1}+c_{2}t^{k+2}$.
Find the remaining two coefficients:%
\begin{align*}
\frac{k+3}{k+2}c_{1}  & =2coef\left(  S_{k}\left(  t\right)  ,t^{k+1}\right)
+\sum_{j=0}^{k-1}coef\left(  S_{j}\left(  t\right)  ,t^{j+1}\right)  \\
& =-\frac{2\left(  k+2\right)  \left(  k+3\right)  }{\left(  2k+3\right)
\left(  2k+4\right)  }\frac{\left(  -1-k\right)  _{k}}{\left(  -2-2k\right)
_{k}}-\sum_{j=0}^{k-1}\frac{\left(  k+2\right)  \left(  k+3\right)  \left(
-1-j\right)  _{j}}{\left(  k+3+j\right)  \left(  k+4+j\right)  \left(
-2-j-k\right)  _{j}}\\
& =-\frac{\left(  k+3\right)  \left(  2\right)  _{k}}{\left(  2k+3\right)
\left(  k+3\right)  _{k}}+\left(  k+3\right)  \left\{  \frac{\left(  2\right)
_{k}}{\left(  k+3\right)  _{k+1}}-\frac{1}{k+3}\right\}  =-1.
\end{align*}
using formula (\ref{sum1}) with $u=-1$ and $m=k-1$.\ Similarly%
\begin{align*}
\frac{k+3}{k+2}c_{2}  & =2coef\left(  S_{k}\left(  t\right)  ,t^{k+2}\right)
+\sum_{j=0}^{k-1}coef\left(  S_{j}\left(  t\right)  ,t^{j+2}\right)  \\
& =-\frac{2\left(  k+2\right)  \left(  k+3\right)  }{\left(  2k+3\right)
\left(  2k+4\right)  }\frac{\left(  -1-k\right)  _{k+1}}{\left(  -2-2k\right)
_{k+1}}-\sum_{j=0}^{k-1}\frac{\left(  k+2\right)  \left(  k+3\right)  \left(
-1-j\right)  _{j+1}}{\left(  k+3+j\right)  \left(  k+4+j\right)  \left(
-2-j-k\right)  _{j+1}}\\
& =-\frac{\left(  k+3\right)  \left(  1\right)  _{k+1}}{\left(  2k+3\right)
\left(  k+2\right)  _{k+1}}-\sum_{j=0}^{k-1}\frac{\left(  k+3\right)  \left(
-1-j\right)  _{j}}{\left(  k+3+j\right)  \left(  k+4+j\right)  \left(
-2-j-k\right)  _{j}}\\
& =-\frac{\left(  k+3\right)  \left(  1\right)  _{k+1}}{\left(  2k+3\right)
\left(  k+2\right)  _{k+1}}+\frac{k+3}{k+2}\left\{  \frac{\left(  2\right)
_{k}}{\left(  k+3\right)  _{k+1}}-\frac{1}{k+3}\right\}  =-\frac{1}{\left(
k+2\right)  }.
\end{align*}
We used $\left(  -1-j\right)  _{j+1}=\left(  -1\right)  \left(  -1-j\right)
_{j}$ and $\left(  -2-j-k\right)  _{j+1}=\left(  -1\right)  \left(
k+2\right)  \left(  -2-j-k\right)  _{j}$ and then the previous sum. Thus
$c_{1}=-\dfrac{k+2}{k+3}$ and $c_{2}=-\dfrac{1}{k+3}$. The factor of $t^{k+1}$
is
\[
-\dfrac{k+2}{k+3}-\dfrac{t}{k+3}=\frac{-k-2-\left(  1-\varepsilon^{2}\right)
}{k+3}=-1+\frac{\varepsilon^{2}}{k+3}%
\]
and%
\[
\chi_{2,k}\left(  \varepsilon^{2}\right)  =\left(  1-\varepsilon^{2}\right)
^{k+1}\left(  -1+\frac{\varepsilon^{2}}{k+3}\right)  +1.
\]
\subsection{Complementary second term of extension to induced measures of master Lovas-Andai formula} \label{CFDintegration}

Let us simplify
\begin{align*}
I\left(  \varepsilon\right)    & :=\frac{\Gamma\left(  1+d+k\right)  ^{2}%
}{\Gamma\left(  \frac{d}{2}\right)  ^{3}\Gamma\left(  1+\frac{d}{2}+k\right)
\Gamma\left(  1+k\right)  }\times\frac{2}{d}\varepsilon^{-d}\\
& \times\int_{0}^{1}dx\int_{\varepsilon^{2}x}^{\varepsilon^{2}}dy[\left\{
\left(  1-x\right)  \left(  1-y\right)  \right\}  ^{k}\left(  xy\right)
^{d/2-1}\left\{  \left(  \varepsilon^{2}-y\right)  \left(  1-x\varepsilon
^{2}\right)  \right\}  ^{d/2}\\
& \times~_{2}F_{1}\left(  -k,\frac{d}{2};1+\frac{d}{2};T\right)  ]\\
T  & :=\frac{\left(  \varepsilon^{2}-y\right)  \left(  1-x\varepsilon
^{2}\right)  }{\left(  1-x\right)  \left(  1-y\right)  \varepsilon^{2}}.
\end{align*}
First we apply the transformation $_{2}F_{1}\left(  a,b;c;t\right)  =\left(
1-t\right)  ^{-a}~_{2}F_{1}(a,c-b;c;\frac{t}{t-1}$), but the series on the
right side only converges if $a$ is a negative integer or $t<\frac{1}{2}$, not
the case in our application, thus \textbf{henceforth assume} $k=0,1,2,3,\ldots
$then%
\begin{align*}
1-T  & =\frac{\left(  1-\varepsilon^{2}\right)  \left(  y-x\varepsilon
^{2}\right)  }{\left(  1-x\right)  \left(  1-y\right)  \varepsilon^{2}},\\
\frac{T}{T-1}  & =-\frac{\left(  \varepsilon^{2}-y\right)  \left(
1-x\varepsilon^{2}\right)  }{\left(  1-\varepsilon^{2}\right)  \left(
y-x\varepsilon^{2}\right)  };
\end{align*}
the integrand becomes%
\begin{align*}
& \left(  xy\right)  ^{d/2-1}\left(  1-\varepsilon^{2}\right)  ^{k}%
\varepsilon^{-2k}\left(  y-x\varepsilon^{2}\right)  ^{k}\left(  1-x\varepsilon
^{2}\right)  ^{d/2}\left(  \varepsilon^{2}-y\right)  _{~}^{d/2}\\
& \times~_{2}F_{1}\left(  -k,1;1+\frac{d}{2};-\frac{\left(  \varepsilon
^{2}-y\right)  \left(  1-x\varepsilon^{2}\right)  }{\left(  1-\varepsilon
^{2}\right)  \left(  y-x\varepsilon^{2}\right)  }\right)  .
\end{align*}
Substitute $y=\varepsilon^{2}u$ so $dy=\varepsilon^{2}du$ and $0\leq x\leq
u\leq1.$ This gives a factor of $\varepsilon^{2d}$ in front of%
\begin{align*}
& \int\limits_{0\leq x\leq u\leq1}\int dx~du~\left(  xu\right)  ^{d/2-1}\\
& \sum_{j=0}^{k}\frac{\left(  -k\right)  _{j}}{\left(  1+\frac{d}{2}\right)
_{j}}\left(  -1\right)  ^{j}\left(  1-x\varepsilon^{2}\right)  ^{d/2+j}\left(
1-u\right)  ^{d/2+j}\left(  1-\varepsilon^{2}\right)  ^{k-j}\left(
u-x\right)  ^{k-j}.
\end{align*}
Isolate the $x$-integral (use the negative binomial series for $\left(
1-x\varepsilon^{2}\right)  ^{d/2+j}$)%
\begin{align*}
& \int_{0}^{u}x^{d/2-1}\left(  u-x\right)  ^{k-j}\sum_{i=0}^{\infty}%
\frac{\left(  -\frac{d}{2}-j\right)  _{i}}{i!}x^{i}\varepsilon^{2i}dx\\
& =\sum_{i=0}^{\infty}\frac{\left(  -\frac{d}{2}-j\right)  _{i}}{i!}%
\frac{\Gamma\left(  \frac{d}{2}+i\right)  \Gamma\left(  k-j+1\right)  }%
{\Gamma\left(  \frac{d}{2}+i+k-j+1\right)  }\varepsilon^{2i}u^{d/2+i+k-j},
\end{align*}
by use of $\int_{0}^{u}x^{\alpha-1}\left(  u-x\right)  ^{\beta-1}%
dx=u^{\alpha+\beta-1}B\left(  \alpha,\beta\right)  $. The inner $u$-integral
is%
\[
\int_{0}^{1}u^{d/2-1}u^{d/2+i+k-j}\left(  1-u\right)  ^{d/2+j}du=\frac
{\Gamma\left(  d+i+k-j\right)  \Gamma\left(  \frac{d}{2}+j+1\right)  }%
{\Gamma\left(  \frac{3d}{2}+i+k+1\right)  }.
\]
Thus the integral is%
\begin{align*}
& \sum_{j=0}^{k}\frac{\left(  -k\right)  _{j}}{\left(  1+\frac{d}{2}\right)
_{j}}\varepsilon^{2d}\left(  1-\varepsilon^{2}\right)  ^{k-j}\left(
-1\right)  ^{j}\\
& \times\sum_{i=0}^{\infty}\frac{\left(  -\frac{d}{2}-j\right)  _{i}}{i!}%
\frac{\Gamma\left(  \frac{d}{2}+i\right)  \Gamma\left(  k-j+1\right)  }%
{\Gamma\left(  \frac{d}{2}+i+k-j+1\right)  }\varepsilon^{2i}\frac
{\Gamma\left(  d+i+k-j\right)  \Gamma\left(  \frac{d}{2}+j+1\right)  }%
{\Gamma\left(  \frac{3d}{2}+i+k+1\right)  }\\
& =\sum_{j=0}^{k}\frac{\left(  -k\right)  _{j}}{\left(  1+\frac{d}{2}\right)
_{j}}\left(  -1\right)  ^{j}\varepsilon^{2d}\left(  1-\varepsilon^{2}\right)
^{k-j}\frac{\Gamma\left(  \frac{d}{2}\right)  \Gamma\left(  k-j+1\right)
\Gamma\left(  d+k-j\right)  \Gamma\left(  \frac{d}{2}+j+1\right)  }%
{\Gamma\left(  \frac{d}{2}+k-j+1\right)  \Gamma\left(  \frac{3d}%
{2}+k+1\right)  }\\
& \times\sum_{i=0}^{\infty}\frac{\left(  -\frac{d}{2}-j\right)  _{i}\left(
\frac{d}{2}\right)  _{i}\left(  d+k-j\right)  _{i}}{i!\left(  \frac{d}%
{2}+k-j+1\right)  _{i}\left(  \frac{3d}{2}+k+1\right)  _{i}}\varepsilon^{2i}.
\end{align*}
The last sum is a $_{3}F_{2}$ with argument $\varepsilon^{2}.$

Simplify the Gamma terms and note $\left(  -k\right)  _{j}=\left(  -1\right)
^{j}\frac{k!}{\left(  k-j\right)  !}$ and $\Gamma\left(  k-j+1\right)
=\left(  k-j\right)  !.$ Then%
\begin{align*}
& \frac{\left(  -k\right)  _{j}}{\left(  1+\frac{d}{2}\right)  _{j}}\left(
-1\right)  ^{j}\frac{\Gamma\left(  \frac{d}{2}\right)  \Gamma\left(
k-j+1\right)  \Gamma\left(  d+k-j\right)  \Gamma\left(  \frac{d}%
{2}+j+1\right)  }{\Gamma\left(  \frac{d}{2}+k-j+1\right)  \Gamma\left(
\frac{3d}{2}+k+1\right)  }\\
& =\frac{k!}{\left(  1+\frac{d}{2}\right)  _{j}}\frac{\Gamma\left(  \frac
{d}{2}\right)  \Gamma\left(  \frac{d}{2}+1\right)  \left(  \frac{d}%
{2}+1\right)  _{j}\Gamma\left(  d\right)  \left(  d\right)  _{k-j}}%
{\Gamma\left(  \frac{d}{2}+1\right)  \left(  \frac{d}{2}+1\right)
_{k-j}\Gamma\left(  \frac{3d}{2}+1+k\right)  }\\
& =\frac{k!\left(  d\right)  _{k-j}\Gamma\left(  \frac{d}{2}\right)
\Gamma\left(  d\right)  }{\left(  \frac{d}{2}+1\right)  _{k-j}\Gamma\left(
\frac{3d}{2}+1+k\right)  }.
\end{align*}
Combine the factors 
\begin{align*}
I\left(  \varepsilon\right)    & =\frac{2\Gamma\left(  1+d+k\right)
^{2}\Gamma\left(  \frac{d}{2}\right)  \Gamma\left(  d\right)  k!\varepsilon
^{d}}{d\Gamma\left(  \frac{d}{2}\right)  ^{3}\Gamma\left(  1+\frac{d}%
{2}+k\right)  \Gamma\left(  1+k\right)  \Gamma\left(  \frac{3d}{2}+1+k\right)
}\\
& \times\sum_{j=0}^{k}\frac{\left(  d\right)  _{k-j}}{\left(  \frac{d}%
{2}+1\right)  _{k-j}}\left(  1-\varepsilon^{2}\right)  ^{k-j}~_{3}F_{2}\left(
%
\genfrac{}{}{0pt}{}{-d/2-j,d/2,d+k-j}{1+d/2+k-j,1+k+3d/2}%
;\varepsilon^{2}\right)  .
\end{align*}
The first line simplifies to%
\[
\frac{\Gamma\left(  1+d+k\right)  ^{2}\Gamma\left(  1+d\right)  }%
{2\Gamma\left(  1+\frac{d}{2}+k\right)  \Gamma\left(  1+\frac{d}{2}\right)
^{2}\Gamma\left(  1+\frac{3d}{2}+k\right)  }\varepsilon^{d}%
\]
(apparently agrees with the postulated $k=0$ expression).
\subsection{Integration problem}
Integration problem, continued, P.S., C.D. 5/23/18, 6/2/18, 6/11/18

Consider the function defined by%
\begin{gather*}
\chi_{d,k}\left(  \varepsilon^{2}\right)  :=\frac{\Gamma\left(  1+d+k\right)
^{2}}{\Gamma\left(  \frac{d}{2}\right)  ^{3}\Gamma\left(  1+\frac{d}%
{2}+k\right)  \Gamma\left(  1+k\right)  }\times\frac{2}{d}\varepsilon^{d}\\
\times\sum_{j=0}^{k}c_{j}\left(  k,\frac{d}{2}\right)  \int_{0}^{1}\int
_{0}^{u}\left(  1-x\varepsilon^{2}\right)  ^{d/2+j}\left(  1-u\right)
^{d/2+j}\left(  1-\varepsilon^{2}\right)  ^{k-j}\left(  u-x\right)
^{k-j}\left(  xu\right)  ^{d/2-1}\mathrm{dx~du},
\end{gather*}
where $c_{k}\left(  k,\frac{d}{2}\right)  =2\frac{k!}{\left(  1+d/2\right)
_{k}}$ and $c_{j}\left(  k,\frac{d}{2}\right)  =\dfrac{\left(  -1\right)
^{j}\left(  -k\right)  _{j}}{\left(  1+\frac{d}{2}\right)  _{j}}$ if $0\leq
j<k$. Furthermore $k$ is a nonnegative integer. In these sections we will show
that  $\chi_{d,k}\left(  \varepsilon^{2}\right)  -1$ is divisible by $\left(
1-\varepsilon^{2}\right)  ^{k+1}$ and when $d$ is an even integer we will
derive formulas for the coefficients of the polynomial $p_{d}\left(
\varepsilon^{2}\right)  $ where $\chi_{d,k}\left(  \varepsilon^{2}\right)
=1+\left(  1-\varepsilon^{2}\right)  ^{k+1}p_{d}\left(  \varepsilon
^{2}\right)  $ and $p_{d}$ is of degree $d-1$ in $\varepsilon^{2}$. The
formula for the coefficient of $\varepsilon^{d+2j}$ is a sum over
approximately $j^{3}/6$ terms, for any value of $k.$

Introduce renamed variables: $d=2a$ and $\varepsilon^{2}=z.$ Change the
summation index to $l=k-j$.%
\begin{align*}
\chi_{2a,k}\left(  z\right)   &  :=\frac{\Gamma\left(  1+2a+k\right)  ^{2}%
}{a\Gamma\left(  a\right)  ^{3}\Gamma\left(  1+a+k\right)  k!}\times\sum
_{l=0}^{k}c_{k-l}\left(  k,a\right)  \\
&  \int_{0}^{1}\int_{0}^{u}z^{a}\left(  1-xz\right)  ^{a+k-l}\left(
1-u\right)  ^{a+k-l}\left(  1-z\right)  ^{l}\left(  u-x\right)  ^{l}\left(
xu\right)  ^{a-1}\mathrm{dx~du}.
\end{align*}

\subsubsection{Derivatives at $z=1$}

We will show that $\chi_{2a,k}\left(  1\right)  =1$ and $\chi_{2a,k}\left(
z\right)  -1$ is divisible by $\left(  1-z\right)  ^{k+1}$, equivalently%
\[
\left(  \frac{d}{dz}\right)  ^{m}\chi_{2a,k}\left(  z\right)  |_{z=1}=0,~1\leq
m\leq k.
\]
There is a trick: carrying out $\left(  \frac{d}{dz}\right)  ^{m}$ and
evaluating the integrand at $z=1$ results in an expression symmetric in $x,u$.
[ Suppose $f\left(  x,u\right)  =f\left(  u,x\right)  $ then $\int_{0}%
^{1}du\int_{0}^{u}f\left(  x,u\right)  dx=\frac{1}{2}\int_{0}^{1}du\int
_{0}^{1}f\left(  x,u\right)  dx$ -- the region of integration is the half of
the unit square under the main diagonal $x=u$.] Then the individual integrals
are ordinary Beta functions. Start with the integral for $\chi_{2a,k}\left(
1\right)  $, only the $l=0$ term remains (observe the symmetry):%
\begin{align*}
&  \int_{0}^{1}\int_{0}^{u}\left(  1-x\right)  ^{a+k}\left(  1-u\right)
^{a+k}\left(  xu\right)  ^{a-1}\mathrm{dx~du}\\
&  =\frac{1}{2}B(a+k+1,a)^{2}=\frac{\Gamma\left(  a+k+1\right)  ^{2}%
\Gamma\left(  a\right)  ^{2}}{2\Gamma\left(  2a+k+1\right)  ^{2}}%
\end{align*}
so%
\begin{align*}
\chi_{2a,k}\left(  1\right)   &  :=\frac{\Gamma\left(  1+2a+k\right)  ^{2}%
}{a\Gamma\left(  a\right)  ^{3}\Gamma\left(  1+a+k\right)  k!}\left(
2\frac{k!}{\left(  1+a\right)  _{k}}\right)  \frac{\Gamma\left(  a+k+1\right)
^{2}\Gamma\left(  a\right)  ^{2}}{2\Gamma\left(  2a+k+1\right)  ^{2}}\\
&  =\frac{\Gamma\left(  a+k+1\right)  }{a\Gamma\left(  a\right)  \left(
1+a\right)  _{k}}=1.
\end{align*}
To differentiate the integrand we use the simple formula (write $\partial
_{z}=\frac{d}{dz}$) (by the product formula $\partial_{z}^{m}\left[  f\left(
z\right)  g\left(  z\right)  \right]  =\sum_{j=0}^{m}\binom{m}{j}\left(
\partial_{z}^{m-j}f\left(  z\right)  \right)  \left(  \partial_{z}^{j}g\left(
z\right)  \right)  $)
\[
\partial_{z}^{m}\left[  f\left(  z\right)  \left(  1-z\right)  ^{l}\right]
|_{z=1}=\left(  -m\right)  _{l}\partial_{z}^{m-l}f\left(  1\right)  .
\]
Apply the formula to $f\left(  z\right)  =z^{a}\left(  1-xz\right)  ^{a+k-l}$:%
\begin{align*}
&  \left(  -m\right)  _{l}\sum_{i=0}^{m-l}\binom{m-l}{i}\left(  -a\right)
_{i}z^{a-i}\left(  -1\right)  ^{i}\left(  -a-k+l\right)  _{m-l-i}%
x^{m-l-i}\left(  1-xz\right)  ^{a+k-m+i}\\
&  =\sum_{i=0}^{m-l}\frac{m!\left(  -1\right)  ^{l+i}}{i!\left(  m-l-i\right)
!}\left(  -a\right)  _{i}\left(  -a-k+l\right)  _{m-l-i}x^{m-l-i}\left(
1-x\right)  ^{a+k-m+i},
\end{align*}
at $z=1$. For simplification let (the third line is a restatement of the
second one)%
\begin{align*}
c_{l}^{\prime}\left(  k,a\right)   &  :=c_{k-l}\left(  k,a\right)
\frac{\left(  1+a\right)  _{k}}{k!},\\
c_{l}^{\prime}\left(  k,a\right)   &  =\frac{\left(  1+a+k-l\right)  _{l}}%
{l!},1\leq l\leq k;~c_{0}^{\prime}\left(  k,a\right)  =2;\\
c_{l}^{\prime}\left(  k,a\right)   &  =\left(  1+\delta_{o,l}\right)  \left(
-1\right)  ^{l}\frac{\left(  -a-k\right)  _{l}}{l!}.
\end{align*}
The integrand is%
\begin{align*}
&  \sum_{l=0}^{m}\left(  1+\delta_{l,0}\right)  \left(  u-x\right)  ^{l}%
\sum_{i=0}^{m-l}\frac{m!\left(  -1\right)  ^{i}}{l!i!\left(  m-l-i\right)
!}\left(  -a\right)  _{i}\left(  -a-k\right)  _{m-i}\\
&  \times x^{m-l-i}\left(  1-x\right)  ^{a+k-m+i}\left(  1-u\right)
^{a+k-l}\left(  xu\right)  ^{a-1}.
\end{align*}
First we show the expression is symmetric in $\left(  u,x\right)  $ by
changing the order of summation and setting $i=m-s$ so that $0\leq l\leq s\leq
m$. We obtain%
\begin{align*}
&  \sum_{s=0}^{m}\left(  -1\right)  ^{m-s}\frac{m!}{\left(  m-s\right)
!s!}\left(  -a\right)  _{m-s}\left(  -a-k\right)  _{s}\left[  \left(
1-x\right)  \left(  1-u\right)  \right]  ^{a+k-s}\left(  xu\right)  ^{a-1}\\
&  \times\sum_{l=0}^{s}\left(  1+\delta_{l,0}\right)  \frac{s!}{l!\left(
s-l\right)  !}x^{s-l}\left(  1-u\right)  ^{s-l}\left(  u-x\right)  ^{l}\\
&  =\sum_{s=0}^{m}\left(  -1\right)  ^{m-s}\frac{m!}{\left(  m-s\right)
!s!}\left(  -a\right)  _{m-s}\left(  -a-k\right)  _{s}\left[  \left(
1-x\right)  \left(  1-u\right)  \right]  ^{a+k-s}\left(  xu\right)  ^{a-1}\\
&  \times\left\{  \left(  x\left(  1-u\right)  +u-x\right)  ^{s}+x^{s}\left(
1-u\right)  ^{s}\right\}  .
\end{align*}
Note the term in $\left\{  \cdot\right\}  $ is $u^{s}\left(  1-x\right)
^{s}+x^{x}\left(  1-u\right)  ^{s}$, symmetric as claimed. Thus the integrand
is%
\begin{align*}
&  \sum_{s=0}^{m}\left(  -1\right)  ^{m-s}\frac{m!}{\left(  m-s\right)
!s!}\left(  -a\right)  _{m-s}\left(  -a-k\right)  _{s}\left(  1-x\right)
^{a+k-s}x^{a+s-1}\left(  1-u\right)  ^{a+k}u^{a-1}\\
&  +\sum_{s=0}^{m}\left(  -1\right)  ^{m-s}\frac{m!}{\left(  m-s\right)
!s!}\left(  -a\right)  _{m-s}\left(  -a-k\right)  _{s}\left(  1-u\right)
^{a+k-s}u^{a+s-1}\left(  1-x\right)  ^{a+k}x^{a-1}.
\end{align*}
It suffices to integrate the first line over the unit square with the result%
\begin{align*}
&  \sum_{s=0}^{m}\left(  -1\right)  ^{m-s}\frac{m!}{\left(  m-s\right)
!s!}\left(  -a\right)  _{m-s}\left(  -a-k\right)  _{s}\frac{\Gamma\left(
a+s\right)  \Gamma\left(  a+k-s+1\right)  \Gamma\left(  a\right)
\Gamma\left(  a+k+1\right)  }{\Gamma\left(  2a+k+1\right)  ^{2}}\\
&  =\frac{\Gamma\left(  a\right)  ^{2}\Gamma\left(  a+k+1\right)  ^{2}}%
{\Gamma\left(  2a+k+1\right)  ^{2}}\left(  -1\right)  ^{m}\sum_{s=0}^{m}%
\frac{m!}{\left(  m-s\right)  !s!}\left(  -a\right)  _{m-s}\left(  a\right)
_{s},
\end{align*}
because $\Gamma\left(  a+s\right)  =\Gamma\left(  a\right)  \left(  a\right)
_{s}$ and $\Gamma\left(  a+k-s+1\right)  =\left(  -1\right)  ^{s}\Gamma\left(
a+k+1\right)  /\left(  -a-k\right)  _{s}$. The sum equals $\left(
-a+a\right)  _{m}=0$ for $1\leq m\leq k$ (this follows from $\sum
\limits_{s=0}^{m}\dfrac{\left(  \alpha\right)  _{m-s}\left(  \beta\right)
_{s}}{\left(  m-s\right)  !s!}=\dfrac{\left(  \alpha+\beta\right)  _{m}}{m!}$,
a version of the Chu-Vandermonde sum). We have shown:

\begin{proposition}
Suppose $a>0$ then $\chi_{2a,k}\left(  1\right)  =1$ and $\left(  \frac{d}%
{dz}\right)  ^{m}\chi_{2a,k}\left(  z\right)  |_{z=1}=0$ for $1\leq m\leq k$,
that is, $\chi_{2a,k}\left(  z\right)  -1$ is divisible by $\left(
1-z\right)  ^{k+1}$.
\end{proposition}

\begin{corollary}
Suppose $a=1,2,3,\ldots$ then $\chi_{2a,k}\left(  z\right)  =1+\left(
1-z\right)  ^{k+1}p_{2a,k}\left(  z\right)  $ where $p_{2a,k}$ is a polynomial
of degree $2a-1$ in $z$.
\end{corollary}

\subsubsection{Coefficients}

Recall%
\begin{align*}
c_{l}^{\prime}\left(  k,a\right)   &  :=c_{k-l}\left(  k,a\right)
\frac{\left(  1+a\right)  _{k}}{k!},\\
c_{l}^{\prime}\left(  k,a\right)   &  =\left(  1+\delta_{0,l}\right)  \left(
-1\right)  ^{l}\frac{\left(  -a-k\right)  _{l}}{l!},0\leq l\leq k.
\end{align*}

It remains to find the coefficients of $z$ in $\left(  \chi_{2a,k}\left(
z\right)  -1\right)  \left(  1-z\right)  ^{-k-1}$. This is extracted from%
\begin{align}
&  \frac{\Gamma\left(  1+2a+k\right)  ^{2}}{\Gamma\left(  a\right)  ^{2}%
\Gamma\left(  1+a+k\right)  ^{2}}\times\nonumber\\
&  \sum_{l=0}^{k}c_{l}^{\prime}\left(  k,a\right)  \int_{0}^{1}\int_{0}%
^{u}z^{a}\left(  1-xz\right)  ^{a+k-l}\left(  1-u\right)  ^{a+k-l}\left(
1-z\right)  ^{l-k-1}\left(  u-x\right)  ^{l}\left(  xu\right)  ^{a-1}%
\mathrm{dx~du}\label{z-int}\\
&  -\left(  1-z\right)  ^{-k-1}.\nonumber
\end{align}
The first line has absorbed the common factor $\frac{k!}{\left(  1+a\right)
_{k}}$. Observe that $z^{i}$ for $i<a$ can occur only in the third line,
indeed the coefficient is $-\dfrac{\left(  k+1\right)  _{i}}{i!}$, so this
takes care of $i<\frac{d}{2}$. The remaining ones require summing over $0\leq
l\leq k$.

The coefficient of $z^{a}$ is $\dfrac{\Gamma\left(  2a+k+1\right)  ^{3}%
}{\Gamma\left(  a+k+1\right)  ^{2}\Gamma\left(  3a+k+1\right)  a!}%
-\dfrac{\left(  k+1\right)  _{a}}{a!}$ for $a=1,2,3,\ldots$, equivalently
$\dfrac{\left(  a+k+1\right)  _{a}^{2}}{a!\left(  2a+k+1\right)  _{a}}%
-\dfrac{\left(  k+1\right)  _{a}}{a!}.$

Recall the rescaled Beta integral: $\int_{0}^{u}x^{\alpha-1}\left(
u-x\right)  ^{\beta-1}\mathrm{dx}=u^{\alpha+\beta-1}B\left(  \alpha
,\beta\right)  $. We begin by computing a basic integral%
\begin{align*}
I_{0}\left(  a,k,l,s\right)   &  :=\int_{0}^{1}\int_{0}^{u}x^{s}\left(
1-u\right)  ^{a+k-l}\left(  u-x\right)  ^{l}\left(  xu\right)  ^{a-1}%
\mathrm{dx~du}\\
&  =B\left(  a+s,l+1\right)  \int_{0}^{1}u^{2a+s+l-1}\left(  1-u\right)
^{a+k-l}\mathrm{du}\\
&  =B\left(  a+s,l+1\right)  B\left(  2a+s+l,a+k-l+1\right) \\
&  =\frac{\Gamma\left(  a+s\right)  l!\Gamma\left(  2a+s+l\right)
\Gamma\left(  a+k-l+1\right)  }{\Gamma\left(  a+s+l+1\right)  \Gamma\left(
3a+k+s+1\right)  }.
\end{align*}
The coefficient of $z^{j}$ in $\left(  1-xz\right)  ^{a+k-l}\left(
1-z\right)  ^{l-k-1}$ is $\sum\limits_{s=0}^{j}\dfrac{\left(  k+1-l\right)
_{j-s}\left(  l-a-k\right)  _{s}}{\left(  j-s\right)  !s!}x^{s}$. Thus the
coefficient of $z^{a+j}$ in (\ref{z-int}) is%
\[
\sum_{l=0}^{k}\sum\limits_{s=0}^{j}c_{l}^{\prime}\left(  k,a\right)
\dfrac{\left(  k+1-l\right)  _{j-s}\left(  l-a-k\right)  _{s}}{\left(
j-s\right)  !s!}I_{0}\left(  a,k,l,s\right)  .
\]
Aiming to sum over $0\leq l\leq k$ we first rewrite the $\left(  l,s\right)  $
term (observe $\Gamma\left(  a+k-l+1\right)  =\Gamma\left(  a+k+1\right)
\left(  -1\right)  ^{l}/\left(  -a-k\right)  _{l}$)
\begin{align*}
&  \frac{\Gamma\left(  a+s\right)  l!\Gamma\left(  2a+s\right)  \Gamma\left(
a+k+1\right)  }{\Gamma\left(  a+s+1\right)  \Gamma\left(  3a+k+s+1\right)
\left(  j-s\right)  !s!}\frac{\left(  2a+s\right)  _{l}}{\left(  a+s+1\right)
_{l}\left(  -a-k\right)  _{l}~l!}\\
&  \times\left(  1+\delta_{0,l}\right)  \left(  -1\right)  ^{l}\left(
k+1-l\right)  _{j-s}\left(  l-a-k\right)  _{s}\left(  1+a+k-l\right)
_{l}\left(  -1\right)  ^{l}%
\end{align*}
and $\left(  1+a+k-l\right)  _{l}\left(  -1\right)  ^{l}=\left(  -a-k\right)
_{l}$, thus the term equals%
\begin{align*}
&  \frac{\Gamma\left(  2a\right)  \Gamma\left(  a+k+1\right)  }{\Gamma\left(
3a+k+1\right)  }\frac{\left(  2a\right)  _{s}}{\left(  a+s\right)  \left(
3a+k+1\right)  _{s}\left(  j-s\right)  !s!}\\
&  \times\frac{\left(  2a+s\right)  _{l}}{\left(  a+s+1\right)  _{l}}\left(
1+\delta_{0,l}\right)  \left(  k+1-l\right)  _{j-s}\left(  l-a-k\right)  _{s}.
\end{align*}
Collect the prefactors (independent of $\left(  s,j\right)  $)%
\[
\frac{\Gamma\left(  1+2a+k\right)  ^{2}}{\Gamma\left(  a\right)  ^{2}%
\Gamma\left(  1+a+k\right)  ^{2}}\frac{\Gamma\left(  2a\right)  \Gamma\left(
a+k+1\right)  }{\Gamma\left(  3a+k+1\right)  }=\frac{\Gamma\left(  2a\right)
\Gamma\left(  1+2a+k\right)  ^{2}}{\Gamma\left(  a\right)  ^{2}\Gamma\left(
1+a+k\right)  \Gamma\left(  3a+k+1\right)  }.
\]
Compute
\begin{align*}
C_{j}\left(  a,k\right)  :=  &  \sum_{s=0}^{j}\frac{\left(  2a\right)  _{s}%
}{\left(  a+s\right)  \left(  3a+k+1\right)  _{s}\left(  j-s\right)  !s!}\\
&  \times\left\{  \sum_{l=0}^{k}\frac{\left(  2a+s\right)  _{l}}{\left(
a+s+1\right)  _{l}}\left(  k+1-l\right)  _{j-s}\left(  l-a-k\right)
_{s}+\left(  k+1\right)  _{j-s}\left(  -a-k\right)  _{s}\right\}  .
\end{align*}
(Observe the doubled $l=0$ term.) The $l$-sum looks hypergeometric but that is
the wrong/useless concept.

\begin{lemma}
\label{easysum}For variables $\alpha,\beta$ and $n=0,1,2,\ldots$%
\[
\sum_{i=0}^{n}\frac{\left(  \alpha\right)  _{i}}{\left(  \beta\right)  _{i}%
}=\frac{1}{\alpha-\beta+1}\left\{  \frac{\left(  \alpha\right)  _{n+1}%
}{\left(  \beta\right)  _{n}}-\beta+1\right\}
\]

\end{lemma}

\begin{proof}
Proceed by induction. For $n=0$ the right hand side is $\frac{1}{\alpha
-\beta+1}\left(  \alpha-\beta+1\right)  =1$. Suppose the formula is true for
$n$ then add $\left(  \alpha\right)  _{n+1}/\left(  \beta\right)  _{n+1}$ to
both sides. The right hand side becomes%
\begin{align*}
&  \frac{1}{\alpha-\beta+1}\frac{\left(  \alpha\right)  _{n+1}}{\left(
\beta\right)  _{n}}+\frac{\left(  \alpha\right)  _{n+1}}{\left(  \beta\right)
_{n+1}}-\frac{\beta-1}{\alpha-\beta+1}\\
&  =\frac{\left(  \alpha\right)  _{n+1}}{\left(  \beta\right)  _{n}}\left\{
\frac{1}{\alpha-\beta+1}+\frac{1}{\beta+n}\right\}  -\frac{\beta-1}%
{\alpha-\beta+1}\\
&  =\frac{\left(  \alpha\right)  _{n+1}}{\left(  \beta\right)  _{n}}\left\{
\frac{\beta+n+\alpha-\beta+1}{\left(  \alpha-\beta+1\right)  \left(
\beta+n\right)  }\right\}  -\frac{\beta-1}{\alpha-\beta+1}=\frac{1}%
{\alpha-\beta+1}\left\{  \frac{\left(  \alpha\right)  _{n+2}}{\left(
\beta\right)  _{n+1}}-\beta+1\right\}  ,
\end{align*}
and this completes the induction.
\end{proof}

The trick to the desired sum is to expand $\left(  k+1-l\right)  _{j-s}\left(
l-a-k\right)  _{s}$ as a sum of $\left(  2\alpha+s+l\right)  _{i}$ with $0\leq
i\leq j$ (this is considering the terms as polynomials in $l$ of degree $\leq
j$).

\begin{lemma}
\label{sum0}Suppose $\alpha,\beta,\lambda,\mu$ are variables and
$n=0,1,2,\ldots,0\leq s\leq j$, then%
\[
\sum_{i=0}^{n}\frac{\left(  \alpha\right)  _{i}}{\left(  \beta\right)  _{i}%
}\left(  \lambda-i\right)  _{j-s}\left(  \mu+i\right)  _{s}=\sum_{u=0}%
^{j}\frac{c_{u}\left(  \alpha,\lambda,\mu,j,s\right)  \left(  \alpha\right)
_{u}}{\alpha+u-\beta+1}\left\{  \frac{\left(  \alpha+u\right)  _{n+1}}{\left(
\beta\right)  _{n}}-\beta+1\right\}  ,
\]
where%
\[
c_{u}\left(  \alpha,\lambda,\mu,j,s\right)  =\sum_{v=0}^{u}\frac{\left(
-1\right)  ^{v}}{v!\left(  u-v\right)  !}\left(  \lambda+\alpha+v\right)
_{j-s}\left(  \mu-\alpha-v\right)  _{s}.
\]

\end{lemma}

\begin{proof}
Suppose the coefficients $c_{u}$ satisfying $\left(  \lambda-i\right)
_{j-s}\left(  \mu+i\right)  _{s}=\sum_{u=0}^{j}c_{u}\left(  \alpha+i\right)
_{u}$ for all $i$ have been determined then
\begin{align*}
\sum_{i=0}^{n}\frac{\left(  \alpha\right)  _{i}}{\left(  \beta\right)  _{i}%
}\left(  \lambda-i\right)  _{j-s}\left(  \mu+i\right)  _{s} &  =\sum_{u=0}%
^{j}c_{u}\frac{\left(  \alpha\right)  _{i}\left(  \alpha+i\right)  _{u}%
}{\left(  \beta\right)  _{i}}=\sum_{u=0}^{j}c_{u}\frac{\left(  \alpha\right)
_{u}\left(  \alpha+u\right)  _{i}}{\left(  \beta\right)  _{i}}\\
&  =\sum_{u=0}^{j}c_{u}\frac{\left(  \alpha\right)  _{u}}{\alpha+u-\beta
+1}\left\{  \frac{\left(  \alpha+u\right)  _{n+1}}{\left(  \beta\right)  _{n}%
}-\beta+1\right\}
\end{align*}
by Lemma \ref{easysum}. Substitute $i=-\alpha-m$ in the equation determining
$\left\{  c_{u}\right\}  $:%
\[
d_{m}:=\left(  \lambda+\alpha+m\right)  _{j-s}\left(  \mu-\alpha-m\right)
_{s}=\sum_{u=0}^{j}c_{u}\left(  -m\right)  _{u}=\sum_{u=0}^{m}\left(
-1\right)  ^{u}\frac{m!}{\left(  m-u\right)  !}c_{u}.
\]
This is a linear transformation from $\left\{  c_{u}\right\}  $ to $\left\{
d_{m}\right\}  $ with a triangular matrix whose inverse is as stated, by a
simple calculation.
\end{proof}

For the specific needed result set
\begin{equation}
\gamma_{u}\left(  j,s\right)  :=\sum_{v=0}^{u}\frac{\left(  -1\right)  ^{v}%
}{v!\left(  u-v\right)  !}\left(  2a+k+1+s+v\right)  _{j-s}\left(
-3a-k-s-v\right)  _{s}\label{gmsum1}%
\end{equation}
by use of the reversal formula $\left(  \alpha\right)  _{n}=\left(  -1\right)
^{n}\left(  1-a-n\right)  _{n}$. Thus for any $l$
\[
\left(  k+1-l\right)  _{j-s}\left(  l-a-k\right)  _{s}=\sum_{u=0}^{j}%
\gamma_{u}\left(  j,s\right)  \left(  2a+s+l\right)  _{u},
\]
in particular, $\left(  k+1\right)  _{j-s}\left(  -a-k\right)  _{s}=\sum
_{u=0}^{j}\gamma_{u}\left(  j,s\right)  \left(  2a+s\right)  _{u}$.

\begin{proposition}
For $j=0,1,2,\ldots$%
\begin{align}
C_{j}\left(  a,k\right)   &  =\sum_{s=0}^{j}\frac{\left(  2a\right)  _{s}%
}{\left(  a+s\right)  \left(  3a+k+1\right)  _{s}\left(  j-s\right)
!s!}\label{Cform1}\\
&  \times\sum_{u=0}^{j}\frac{\gamma_{u}\left(  j,s\right)  \left(
2a+s\right)  _{u}}{a+u}\left\{  \frac{\left(  2a+s+u\right)  _{k+1}}{\left(
a+s+1\right)  _{k}}+u-s\right\}  .\nonumber
\end{align}

\end{proposition}

\begin{proof}
The formula follows from Lemma \ref{sum0} with $\alpha=2a+s$, $\beta=a+s+1$,
$\lambda=k+1$, $\mu=-a-k$. Thus $\beta-1=a+s$ and $\alpha+u-\beta+1=a+u$. Next
add $\left(  k+1\right)  _{j-s}\left(  -a-k\right)  _{s}$ in the sum form
$\sum_{u=0}^{j}\gamma_{u}\left(  j,s\right)  \left(  2a+s\right)  _{u}$ then
$1-\frac{a+s}{a+u}=\frac{u-s}{a+u}$.
\end{proof}

Next we apply a transformation to $\gamma_{u}\left(  j,s\right)  $.

\begin{proposition}
For $0\leq s,u\leq j$%
\begin{align}
\gamma_{u}\left(  j,s\right)    & =\left(  -1\right)  ^{u+s}\frac{\left(
3a+k+1\right)  _{s}\left(  2a+k+1\right)  _{j}\left(  j-s\right)  !}{u!\left(
2a+k+1\right)  _{s+u}}\label{gmsum2}\\
& \times\sum_{i=\max\left(  0,s+u-j\right)  }^{\min\left(  s,u\right)  }%
\frac{\left(  -s\right)  _{i}\left(  -u\right)  _{i}\left(  2a+k+1+j\right)
_{i}}{i!\left(  j-s-u+i\right)  !\left(  3a+k+1\right)  _{i}}\nonumber
\end{align}

\end{proposition}

\begin{proof}
Rewrite%
\begin{align*}
\left(  2a+k+1+s+v\right)  _{j-s}  & =\frac{\left(  2a+k+1\right)  _{j}\left(
2a+k+1+j\right)  _{v}}{\left(  2a+k+1\right)  _{s}\left(  2a+k+1+s\right)
_{v}},\\
\left(  -3a-k-s-v\right)  _{s}  & =\sum_{i=0}^{s}\binom{s}{i}\left(
-3a-k-s\right)  _{s-i}\left(  -v\right)  _{i},
\end{align*}
by the Chu-Vandermonde sum. Thus
\[
\gamma_{u}\left(  j,s\right)  =\frac{\left(  2a+k+1\right)  _{j}}{u!\left(
2a+k+1\right)  _{s}}\sum_{i=0}^{s}\binom{s}{i}\left(  -3a-k-s\right)
_{s-i}\sum_{v=0}^{u}\frac{\left(  -u\right)  _{v}\left(  2a+k+1+j\right)
_{v}}{v!\left(  2a+k+1+s\right)  _{v}}\left(  -v\right)  _{i}.
\]
In the sum over $v$ change the summation variable $v=i+n$ (so that $\left(
-v\right)  _{i}=\left(  -1\right)  ^{i}v!/n!$) to obtain
\begin{gather*}
\left(  -1\right)  ^{i}\sum_{n=0}^{u-i}\frac{\left(  -u\right)  _{i+n}\left(
2a+k+1+j\right)  _{i+n}}{n!\left(  2a+k+1+s\right)  _{i+n}}\\
=\left(  -1\right)  ^{i}\frac{\left(  -u\right)  _{i}\left(  2a+k+1+j\right)
_{i}}{\left(  2a+k+1+s\right)  _{i}}\times\sum_{n=0}^{u-i}\frac{\left(
i-u\right)  _{n}\left(  2a+k+1+i+j\right)  _{n}}{n!\left(  2a+k+1+i+s\right)
_{n}}\\
=\left(  -1\right)  ^{i}\frac{\left(  -u\right)  _{i}\left(  2a+k+1+j\right)
_{i}\left(  s-j\right)  _{u-i}}{\left(  2a+k+1+s\right)  _{i}\left(
2a+k+1+i+s\right)  _{u-i}}=\left(  -1\right)  ^{i}\frac{\left(  -u\right)
_{i}\left(  2a+k+1+j\right)  _{i}\left(  s-j\right)  _{u-i}}{\left(
2a+k+1+s\right)  _{u}}.
\end{gather*}
Also
\[
\left(  -3a-k-s\right)  _{s-i}=\left(  -1\right)  ^{s-i}\left(
3a+k+1+i\right)  _{s-i}=\left(  -1\right)  ^{s-i}\frac{\left(  3a+k+1\right)
_{s}}{\left(  3a+k+1\right)  _{i}}%
\]
and $\left(  s-j\right)  _{u-i}=\left(  -1\right)  ^{u-i}\frac{\left(
j-s\right)  !}{\left(  j-s-u+i\right)  !}$; note the terms with $i>u$ vanish
due to the factor $\left(  -v\right)  _{i}$ and $v\leq u$. Also $\binom{s}%
{i}=\left(  -1\right)  ^{i}\frac{\left(  -s\right)  _{i}}{i!}$ and the powers
of $\left(  -1\right)  $ add up to $s+u$. Thus%
\[
\gamma_{u}\left(  j,s\right)  =\left(  -1\right)  ^{u+s}\frac{\left(
3a+k+1\right)  _{s}\left(  2a+k+1\right)  _{j}\left(  j-s\right)  !}{u!\left(
2a+k+1\right)  _{s+u}}\sum_{i=0}^{s}\frac{\left(  -s\right)  _{i}\left(
-u\right)  _{i}\left(  2a+k+1+j\right)  _{i}}{i!\left(  j-s-u+i\right)
!\left(  3a+k+1\right)  _{i}};
\]
in fact the sum is over $\max\left(  0,s+u-j\right)  \leq i\leq\min\left(
s,u\right)  $. The terms for $i$ outside this interval vanish.
\end{proof}

\begin{proposition}
For $j=0,1,2\ldots$%
\begin{align}
C_{j}\left(  a,k\right)    & =\sum_{s=0}^{j}\sum_{u=0}^{j}\frac{\left(
-1\right)  ^{u+s}\left(  2a\right)  _{s+u}\left(  2a+k+1\right)  _{j}}{\left(
a+s\right)  \left(  a+u\right)  s!u!\left(  2a+k+1\right)  _{s+u}}%
\frac{\left(  2a+s+u\right)  _{k+1}}{\left(  a+s+1\right)  _{k}}%
\label{cjak}\\
& \times\sum_{i=\max\left(  0,s+u-j\right)  }^{\min\left(  s,u\right)  }%
\frac{\left(  -s\right)  _{i}\left(  -u\right)  _{i}\left(  2a+k+1+j\right)
_{i}}{i!\left(  j-s-u+i\right)  !\left(  3a+k+1\right)  _{i}}.\nonumber
\end{align}

\end{proposition}

\begin{proof}
Combine the factors in (\ref{Cform1}) with (\ref{gmsum2}) and cancel the
factors $\left(  3a+k+1\right)  _{s}$ and $\left(  j-s\right)  !.$The part due
to the term $\left(  u-s\right)  $ in $\left\{  \cdot\right\}  $ adds up to
zero, summed over $0\leq s,u\leq j$  because the factor in front of $\left\{
\cdot\right\}  $ is symmetric in $\left(  u,s\right)  $.
\end{proof}

There may be another simplification of the triple sum; nevertheless it is an
easy symbolic computation for reasonably small values of $j$.

\subsubsection{Summary and examples}

To summarize, the coefficient of $z^{a+j}$ in $\left(  \chi_{2a,k}\left(
z\right)  -1\right)  \left(  1-z\right)  ^{-k-1}$ is%
\begin{align*}
& -\frac{\left(  k+1\right)  _{a+j}}{\left(  a+j\right)  !}+\frac
{a\Gamma\left(  2a+k+1\right)  ^{2}\Gamma\left(  2a+k+1+j\right)  }%
{\Gamma\left(  a\right)  \Gamma\left(  a+k+1\right)  ^{2}\Gamma\left(
3a+k+1\right)  }\times\\
& \sum_{s=0}^{j}\sum_{u=0}^{j}\frac{\left(  -1\right)  ^{u+s}\left(
a+1\right)  _{s}}{\left(  a+s\right)  \left(  a+u\right)  s!u!\left(
a+k+1\right)  _{s}}\sum_{i=\max\left(  0,s+u-j\right)  }^{\min\left(
s,u\right)  }\frac{\left(  -s\right)  _{i}\left(  -u\right)  _{i}\left(
2a+k+1+j\right)  _{i}}{i!\left(  j-s-u+i\right)  !\left(  3a+k+1\right)  _{i}}%
\end{align*}

Here is the calculation for this formula which combines the prefactor with
(\ref{cjak}):%
\begin{align*}
& \frac{\Gamma\left(  2a\right)  \Gamma\left(  1+2a+k\right)  ^{2}\left(
2a\right)  _{s+u}\left(  2a+k+1\right)  _{j}\left(  2a+s+u\right)  _{k+1}%
}{\Gamma\left(  a\right)  ^{2}\Gamma\left(  1+a+k\right)  \Gamma\left(
3a+k+1\right)  \left(  2a+k+1\right)  _{s+u}\left(  a+s+1\right)  _{k}}\\
& =\frac{\Gamma\left(  1+2a+k\right)  ^{2}\Gamma\left(  2a\right)  \left(
2a\right)  _{s+u+k+1}\left(  2a+k+1\right)  _{j}\left(  a\right)  _{s+1}%
}{\Gamma\left(  a\right)  ^{2}\Gamma\left(  1+a+k\right)  \Gamma\left(
3a+k+1\right)  \left(  2a+k+1\right)  _{s+u}\left(  a\right)  _{k+1+s}}\\
& =\frac{\Gamma\left(  1+2a+k\right)  ^{2}\Gamma\left(  2a+k+1+j\right)
}{\Gamma\left(  a\right)  \Gamma\left(  a+k+1\right)  \Gamma\left(
3a+k+1\right)  }\times\frac{\left(  a\right)  _{s+1}}{\left(  a+k+1\right)
_{s}},
\end{align*}
using $\Gamma\left(  2a\right)  \left(  2a\right)  _{k+1=s+u}=\Gamma\left(
2a+k+1\right)  \left(  2a+k+1\right)  _{s+u}$ and $\Gamma\left(
2a+k+1\right)  \left(  2a+k+1\right)  _{j}=\Gamma\left(  2a+k+1+j\right)  $.

As illustration here are the formulas for $\chi_{2,k},\chi_{4.k}$ and
$\chi_{6,k}$:%
\begin{align*}
\chi_{2,k}\left(  z\right)   &  =1+\left(  1-z\right)  ^{k+1}\left(
-1+\frac{1}{k+3}z\right)  ,\\
\chi_{4.k}\left(  z\right)   &  =1+\left(  1-z\right)  ^{k+1}\left(
-1-\left(  k+1\right)  z+\frac{2\left(  2k^{2}+14k+21\right)  }{\left(
k+5\right)  \left(  k+6\right)  }z^{2}-\frac{6\left(  k+3\right)  }{\left(
k+6\right)  \left(  k+7\right)  }z^{3}\right)  ,\\
\chi_{6,k}\left(  z\right)   &  =1+\left(  1-z\right)  ^{k+1}\{-1-\left(
k+1\right)  z-\frac{\left(  k+1\right)  \left(  k+2\right)  }{2}z^{2}\\
&  +\frac{3\left(  3k^{4}+60k^{3}+432k^{2}+1230k+1264\right)  }{2\left(
k+7\right)  \left(  k+8\right)  \left(  k+9\right)  }z^{3}-\frac{6\left(
k+4)(3k^{2}+33k+80\right)  }{\left(  k+8\right)  \left(  k+9\right)  \left(
k+10\right)  }z^{4}\\
&  +\frac{30\left(  k+4\right)  \left(  k+5\right)  }{\left(  k+9\right)
\left(  k+10\right)  \left(  k+11\right)  }z^{5}\}.
\end{align*}
\subsubsection{Comments on the last steps using $\chi_{d,k}\left(  \varepsilon^{2}\right)  $} \label{laststeps}

\mbox{Derivation of the integral formula}

Besides some Gamma factors the task is to integrate over $0<\varepsilon<1$%

\begin{align*}
&  2\varepsilon^{-3-4d-2k}\left(  1-\varepsilon^{2}\right)  ^{d}\chi
_{d,k}\left(  \varepsilon^{2}\right)  ~_{2}F_{1}\left(
\genfrac{}{}{0pt}{}{2+3d+2k,2+3d+2k}{4+6d+4k}%
;1-\varepsilon^{-2}\right) \\
&  =2\varepsilon^{1+2d+2k}\left(  1-\varepsilon^{2}\right)  ^{d}\chi
_{d,k}\left(  \varepsilon^{2}\right)  ~_{2}F_{1}\left(
\genfrac{}{}{0pt}{}{2+3d+2k,2+3d+2k}{4+6d+4k}%
;1-\varepsilon^{2}\right)
\end{align*}
then divide by
\[
_{3}F_{2}\left(
\genfrac{}{}{0pt}{}{2+3d+2k,2+3d+2k,1+d}{4+6d+4k,2+2d+k}%
;1\right)
\]
First change variables $z=\varepsilon^{2},\mathrm{d}z=2\varepsilon
~\mathrm{d}\varepsilon$. Use the integral representation%
\[
_{2}F_{1}\left(
\genfrac{}{}{0pt}{}{\alpha,\alpha}{2\alpha}%
,1-z\right)  =\frac{\Gamma\left(  2\alpha\right)  }{\Gamma\left(
\alpha\right)  ^{2}}\int_{0}^{1}\left(  t\left(  1-t\right)  \right)
^{\alpha-1}\left(  1-\left(  1-z\right)  t\right)  ^{-\alpha}dt,
\]
then integrate with respect to $z$ first, (doing $t$ first results in $\log z$
terms). Again disregarding Gamma factors compute%
\begin{equation}
I_{d,k}:=\int_{0}^{1}\int_{0}^{1}z^{d+k}\chi_{d,k}\left(  z\right)  \left(
1-z\right)  ^{d}\left(  1-\left(  1-z\right)  t\right)  ^{-2-3d-2k}%
\mathrm{d}z\left\{  t\left(  1-t\right)  \right\}  ^{1+3d+2k}\mathrm{d}t.
\label{int(d,k)}%
\end{equation}
In this order of integration it is a relatively straightforward task for
symbolic computation, the result of the $z$-integral is a polynomial in $t$
divided by $\left(  1-t\right)  ^{3d/2+k+1}$. More details for the evaluation
are given in the following section.

For the denominator $_{3}F_{2}$-series use a transformation to produce a
terminating sum with $d+k+1$ terms. Note that $2+3d+2k=(2+2d+k)+(d+k)$ . There
is a transformation (for $n=0,1,2,\ldots$and $\gamma-\alpha-\beta>n$)%
\[
_{3}F_{2}\left(
\genfrac{}{}{0pt}{}{\alpha,\beta,\delta+n}{\gamma,\delta}%
;1\right)  =\frac{\Gamma\left(  \gamma\right)  \Gamma\left(  \gamma
-\alpha-\beta\right)  }{\Gamma\left(  \gamma-\alpha\right)  \Gamma\left(
\gamma-\beta\right)  }~_{3}F_{2}\left(
\genfrac{}{}{0pt}{}{-n,\alpha,\beta}{1+\alpha+\beta-\gamma,\delta}%
;1\right)  .
\]
Let%
\[
S_{d,k}:=~_{3}F_{2}\left(
\genfrac{}{}{0pt}{}{-d-k,2+3d+2k,1+d}{2+2d+k,-2d-2k}%
;1\right)
\]
(not \textit{regularized} !). The previously ignored Gamma factors are now
collected -the probability is%
\[
\frac{\Gamma\left(  2+2d+k\right)  \Gamma\left(  3+5d+4k\right)  }%
{\Gamma\left(  2+3d+2k\right)  \Gamma\left(  1+d+k\right)  \Gamma\left(
1+2d+2k\right)  d!}\frac{I_{d,k}}{S_{d,k}}%
\]

\mbox{Evaluation of the double integral}

Throughout $a=d/2$ is a positive integer and $k=0,1,2,\ldots$.

Consider the integrand%
\[
\frac{z^{2a+k}\chi_{2a,k}\left(  z\right)  \left(  1-z\right)  ^{2a}}{\left(
1-\left(  1-z\right)  t\right)  ^{2+6a+2k}}.
\]
The numerator is a polynomial in $z$ of degree $\left(  2a+k\right)
+(2a+k)+2a=6a+2k$. Also from a previous formula (the key integral) the lowest
nonzero power of $z$ in $\chi_{2a,k}\left(  z\right)  $ is $a$; thus the
numerator vanishes at $z=0$ of order $3a+k-1$. First we solve the following
problem: determine $q\left(  z;t\right)  $ such that%
\[
\frac{d}{dz}\frac{q\left(  z;t\right)  }{\left(  1-\left(  1-z\right)
t\right)  ^{m-1}}=\frac{p\left(  z\right)  }{\left(  1-\left(  1-z\right)
t\right)  ^{m}},
\]
where $p\left(  z\right)  $ is a given polynomial of degree $m-2$ and
$q\left(  z;t\right)  $ is polynomial in $z$. By the product rule%
\begin{align}
\frac{p\left(  z\right)  }{\left(  1-\left(  1-z\right)  t\right)  ^{m}}  &
=\frac{\frac{d}{dz}q\left(  z;t\right)  }{\left(  1-t+zt\right)  ^{m-1}}%
-\frac{\left(  m-1\right)  tq\left(  z;t\right)  }{\left(  1-t+zt\right)
^{m}}\nonumber\\
&  =\frac{\left(  1-t+zt\right)  \frac{d}{dz}q\left(  z;t\right)  -\left(
m-1\right)  tq\left(  z;t\right)  }{\left(  1-t+zt\right)  ^{m}}\nonumber\\
p\left(  z\right)   &  =\left(  1-t+zt\right)  \frac{d}{dz}q\left(
z;t\right)  -\left(  m-1\right)  tq\left(  z;t\right)  \label{eqnpq}%
\end{align}
Let $p\left(  z\right)  =\sum\limits_{i=0}^{m-2}\beta_{i}z^{i}$ and $q\left(
z;t\right)  =\sum\limits_{i=0}^{m-2}\gamma_{i}\left(  t\right)  z^{i}$. The
coefficient of $z^{j}$ in the numerator is%
\[
\beta_{j}=t\left(  j-m+1\right)  \gamma_{j}\left(  t\right)  +\left(
1-t\right)  \left(  j+1\right)  \gamma_{j+1}\left(  t\right)  ;
\]
in the special case $j=m-2$ one has $\gamma_{m-2}=-\frac{1}{t}\beta_{m-2}$.
The solution $\left[  \gamma_{j}\right]  _{j=0}^{m-2}$ is found recursively
starting with $j=m-3$ down to $j=0$:%
\begin{equation}
\gamma_{j}\left(  t\right)  =\frac{1}{t\left(  m+1-j\right)  }\left\{  \left(
j+1\right)  \left(  1-t\right)  \gamma_{j+1}\left(  t\right)  -\beta
_{j}\right\}  . \label{receqgm}%
\end{equation}
Thus $\gamma_{j}\left(  t\right)  $ is a polynomial in $\frac{1}{t}$ with the
highest power of $\frac{1}{t}$ being $m-1-j$ and the lowest being $1$ (no
constant term since $t\gamma_{m-2}=-\beta_{m-2}$). Let $\widetilde{\gamma
}\left(  t\right)  :=t^{m-1}\sum\limits_{j=0}^{m-2}\gamma_{j}\left(  t\right)
$ and $\widetilde{\gamma_{0}}\left(  t\right)  :=t^{m-1}\gamma_{0}\left(
t\right)  $; then both $\widetilde{\gamma}\left(  t\right)  $ and
$\widetilde{\gamma_{0}}\left(  t\right)  $ are polynomials in $t$ of degree
$\leq m-2$ (from $\left(  \frac{1}{t}\right)  /t^{m-1}=t^{m-2}$),

\begin{lemma}
Suppose that $\beta_{i}=0$ for $0\leq i<n$ for some $n$ then $\gamma_{j}$ is
divisible by $\left(  1-t\right)  ^{n-j}$ for $j=0,1,\ldots,n-1$.
\end{lemma}

\begin{proof}
By formula (\ref{receqgm}) $\gamma_{n-1}\left(  t\right)  =\dfrac{n}{t\left(
m-n+2\right)  }\left(  1-t\right)  \gamma_{n}\left(  t\right)  $ (since
$\beta_{n-1}=0$); thus $\gamma_{n-2}\left(  t\right)  =\dfrac{n\left(
n-1\right)  }{t^{2}\left(  m-n+2\right)  \left(  m-n+3\right)  }\left(
1-t\right)  ^{2}\gamma_{n}\left(  t\right)  $ and so on (the proof is
completed by induction).
\end{proof}

The previous hypothesis refers to the order of vanishing at $z=0$ of $p\left(
z\right)  $. The following deals with the corresponding property at $z=1$.

\begin{lemma}
\label{dz(p(1))}Suppose $\left(  \frac{d}{dz}\right)  ^{i}p\left(  z\right)
=0|_{z=1}$ for $0\leq i<r$ then $\deg\left(  \widetilde{\gamma}\left(
t\right)  \right)  \leq m-2-r$.
\end{lemma}

\begin{proof}
Differentiate equation (\ref{eqnpq}) $s$ times to obtain%
\[
\left(  \frac{d}{dz}\right)  ^{s}p\left(  z\right)  =\left(  s-m+1\right)
t\left(  \frac{d}{dz}\right)  ^{s}q\left(  z;t\right)  +\left(  1-t+zt\right)
\left(  \frac{d}{dz}\right)  ^{s+1}q\left(  z;t\right)  .
\]
Suppose $0\leq s<r$ and $z=1$ then%
\[
\left(  \frac{d}{dz}\right)  ^{s}q\left(  1;t\right)  =\frac{1}{\left(
m-1-s\right)  t}\left(  \frac{d}{dz}\right)  ^{s+1}q\left(  1;t\right)  .
\]
By induction%
\[
q\left(  1;t\right)  =\left(  \frac{1}{t}\right)  ^{r}\frac{\left(
m-r-1\right)  !}{\left(  m-1\right)  !}\left(  \frac{d}{dz}\right)
^{r}q\left(  1;t\right)  .
\]
Also $\left(  \frac{d}{dz}\right)  ^{r}q\left(  1;t\right)  =\sum
\limits_{j=r}^{m-2}\dfrac{j!}{\left(  j-r\right)  !}\gamma_{j}\left(
t\right)  $ so the lowest power of$\frac{1}{t}$ in $\left(  \frac{d}%
{dz}\right)  ^{r}q\left(  1;t\right)  $ is $1$. Thus $q\left(  1;t\right)  $
is a polynomial in $\frac{1}{t}$ of degree $\leq m-1$ and the lowest power of
$\frac{1}{t}$ is $r+1$. Hence the highest power of $t$ in $\widetilde{\gamma
}\left(  t\right)  :=t^{m-1}q\left(  1;t\right)  $ is $t^{m-2-r}$.
\end{proof}

\begin{proposition}
\label{int(p,m)}Suppose $\beta_{i}=0$ for $0\leq i<n$ then the integral%
\[
I\left(  p,m\right)  :=\int_{0}^{1}\frac{p\left(  z\right)  }{\left(
1-\left(  1-z\right)  t\right)  ^{m}}\mathrm{d}z=\frac{h\left(  t\right)
}{\left(  1-t\right)  ^{m-n-1}}%
\]
where $h\left(  t\right)  $ is a polynomial in $t$ and $\deg h\left(
t\right)  \leq m-n-2$.
\end{proposition}

\begin{proof}
By elementary calculus%
\[
I\left(  p,m\right)  =\sum_{j=0}^{m-2}\gamma_{j}\left(  t\right)
-\frac{\gamma_{0}\left(  t\right)  }{\left(  1-t\right)  ^{m-1}}%
=\frac{\widetilde{\gamma}\left(  t\right)  \left(  1-t\right)  ^{m-1}%
-\widetilde{\gamma_{0}}\left(  t\right)  }{t^{m-1}\left(  1-t\right)  ^{m-1}}%
\]
By hypothesis $\widetilde{\gamma_{0}}\left(  t\right)  =\left(  1-t\right)
^{n}\sum_{i=0}^{m-n-2}d_{i}t^{i}$ for some coefficients $\left\{
d_{i}\right\}  $ then
\[
I\left(  p,m\right)  =\frac{\widetilde{\gamma}\left(  t\right)  \left(
1-t\right)  ^{m-n-1}-\sum_{i=0}^{m-n-2}d_{i}t^{i}}{t^{m-1}\left(  1-t\right)
^{m-n-1}}%
\]
Suppose $\widetilde{\gamma}\left(  t\right)  \left(  1-t\right)  ^{m-n-1}%
=\sum\limits_{i=0}^{2m-n-3}c_{i}t^{i}$ for some coefficients $\left\{
c_{i}\right\}  $ then $\widetilde{\gamma}\left(  t\right)  \left(  1-t\right)
^{m-n-1}-\sum\limits_{i=0}^{m-n-2}d_{i}t^{i}=\sum\limits_{i=m-n-1}^{2m-3}%
c_{i}t^{i}+\sum\limits_{i=0}^{m-n-2}\left(  c_{i}-d_{i}\right)  t^{i}$ but the
integrand is analytic in $t$ for $\left\vert t\right\vert <1$ thus $t^{m-1}$
must be a factor of the numerator and it follows that $c_{i}=d_{i}$ for $0\leq
i\leq m-n-2$ and $c_{i}=0$ for $m-n-1\leq i<m-1$ so that%
\[
I\left(  p,m\right)  =\frac{\sum_{j=0}^{m-n-2}c_{j+m-1}t^{j}}{\left(
1-t\right)  ^{m-n-1}}.
\]
This completes the proof.
\end{proof}

\begin{corollary}
\label{int(p,m)cor}If also $\left(  \frac{d}{dz}\right)  ^{i}p\left(
z\right)  |_{z=1}=0$ for $0\leq i<r$ then $\deg h\left(  t\right)  \leq
m-n-r-2$.
\end{corollary}

\begin{proof}
By Lemma \ref{dz(p(1))} $\deg\widetilde{\gamma}\left(  t\right)  \leq m-2-r$
thus $\deg\left(  \widetilde{\gamma}\left(  t\right)  \left(  1-t\right)
^{m-n-1}\right)  \leq2m-n-r-3$ implying $c_{j+m-1}=0$ for $j>m-n-r-2$.
\end{proof}

Suppose $\beta_{i}=0$ for $0\leq i<n$, $\left(  \frac{d}{dz}\right)
^{j}p\left(  z\right)  |_{z=1}=0$ for $0\leq j<r$ and $\widetilde{\gamma
}\left(  t\right)  =\sum\limits_{j=0}^{m-2-r}g_{j}t^{j}$ then the coefficients
$\left[  c_{j}\right]  _{j=m-1}^{2m-n-r-3}$ depend only on $\left[
g_{i}\right]  _{i=n}^{m-r-2}$ (because $\deg\left(  \left(  1-t\right)
^{m-n-1}\sum_{i=0}^{n-1}g_{i}t^{i}\right)  \leq m-2$. Indeed%
\[
c_{j+m-1}=\sum_{i=0}^{m-n-r-2-j}\left(  -1\right)  ^{m-n-1-i}\binom{m-n-1}%
{i}g_{n+j+i};
\]
note only the terms $g_{u}$ with $u\geq n$ appear; also $c_{2m-n-r-3}=\left(
-1\right)  ^{m-n-1}g_{m-r-2}$.

We now apply the general formulas to $p\left(  z\right)  =z^{2a+k}\chi
_{2a,k}\left(  z\right)  \left(  1-z\right)  ^{2a}$ and $m=6a+2k+2$. As
mentioned above $\left(  \frac{d}{dz}\right)  ^{i}p\left(  z\right)
|_{z=0}=0$ for $0\leq i<3a+k$ and $\left(  \frac{d}{dz}\right)  ^{j}p\left(
z\right)  |_{z=1}=0$ for $0\leq j<2a$. By Proposition \ref{int(p,m)} and
Corollary \ref{int(p,m)cor}%
\[
\int_{0}^{1}\frac{z^{2a+k}\chi_{2a,k}\left(  z\right)  \left(  1-z\right)
^{2a}}{\left(  1-\left(  1-z\right)  t\right)  ^{2+6a+2k}}\mathrm{d}%
z=\frac{h\left(  t\right)  }{\left(  1-t\right)  ^{3a+k+1}},
\]
where $h\left(  t\right)  $ is a polynomial in $t$ and $\deg\left(  h\left(
t\right)  \right)  =a+k$. To finish the integration with respect to $\left\{
t\left(  1-t\right)  \right\}  ^{1+6a+2k}\mathrm{d}t$ let $h\left(  t\right)
=\sum_{i=0}^{a+k}b_{i}t^{i}$ and compute
\begin{gather*}
\int_{0}^{1}\frac{h\left(  t\right)  }{\left(  1-t\right)  ^{3a+k+1}}\left\{
t\left(  1-t\right)  \right\}  ^{1+6a+2k}\mathrm{d}t=\int_{0}^{1}\sum
_{i=0}^{a+k}b_{i}t^{i}\left\{  t^{1+6a+2k}\left(  1-t\right)  ^{3a+k}\right\}
\mathrm{d}t\\
=\frac{\Gamma\left(  6a+2k+2\right)  \Gamma\left(  3a+k+1\right)  }%
{\Gamma\left(  9a+3k+3\right)  }\sum_{i=0}^{a+k}b_{i}\frac{\left(
6a+2k+2\right)  _{i}}{\left(  9a+3k+3\right)  _{i}}.
\end{gather*}
This used the Beta integral $\int_{0}^{1}t^{i}t^{M-1}\left(  1-t\right)
^{N-1}\mathrm{d}t=\frac{\Gamma\left(  M\right)  \Gamma\left(  N\right)
}{\Gamma\left(  M+N\right)  }\frac{\left(  M\right)  _{i}}{\left(  M+N\right)
_{i}}$. The integral was denoted $I_{2a,k}$ in formula (\ref{int(d,k)}) in the
previous section.

We have demonstrated a computational scheme that involves only finite
operations; basically just some elementary polynomial algebra.

\begin{acknowledgements}
My considerable thanks to Charles Dunkl for contributing the  appendices, and his general support and advice.
\end{acknowledgements}

\bibliography{main}

\end{document}